\documentclass{amsart}
\usepackage{amssymb}
\usepackage{amsmath}
\usepackage{amsthm}
\usepackage{mathrsfs}
\usepackage{colortbl}
\newtheorem{theorem}{Theorem}
\newtheorem{lemma}[theorem]{Lemma}
\newtheorem{assumption}[theorem]{Assumption}
\newtheorem{remark}[theorem]{Remark}

\theoremstyle{definition}
\newtheorem{definition}[theorem]{Definition}

\newtheorem{proposition}[theorem]{Proposition}

\makeatletter

\@addtoreset{theorem}{section}
\makeatother

\makeatletter

\@addtoreset{equation}{section}
\makeatother

\begin{document}                                                 
\title[On Compressible Fluid Flows on an Evolving Surface with a Boundary]{On Generalized Compressible Fluid Systems on an Evolving Surface with a Boundary}                                
\author[Hajime Koba]{Hajime Koba}                                
\address{Graduate School of Engineering Science, Osaka University,\\
1-3 Machikaneyamacho, Toyonaka, Osaka, 560-8531, Japan}                                  
\email{iti@sigmath.es.osaka-u.ac.jp}

\begin{abstract}
We consider compressible fluid flow on an evolving surface with a smooth boundary from an energetic point of view. We employ both an energetic variational approach and the first law of thermodynamics to make mathematical models for compressible fluid flow on the evolving surface. Moreover, we investigate the boundary conditions in co-normal direction for our fluid systems to study the conservation and energy laws of the systems.
\end{abstract}       
\maketitle

\section{Introduction}\label{sect1}

We are interested in a mathematical modeling of soap bubble clusters and membrane flows. A soap bubble cluster consists of some surfaces with a boundary. When we focus on a soap bubble, we can see the fluid flow in the bubble. We call the fluid flow in the bubble a surface flow. We can consider a surface flow as fluid flow on an evolving surface. An evolving surface means that the surface is moving or the shape of the surface is changing along with the time. This paper has two aims. The first one is to derive the generalized compressible fluid system on an evolving surface with a boundary from an energetic point of view. \emph{Generalized} means that the system includes Newtonian fluid, non-Newtonian fluid, and power law fluid. The second one is to investigate the natural boundary conditions for the system. \emph{Natural} means that the system satisfies some conservation laws.

Some researchers such as Arnol'd (\cite{Arn66} , \cite{Arn97}), Taylor \cite{Tay92}, Dziuk-Elliott \cite{DE07}, Koba-Liu-Giga \cite{KLG17}, Koba \cite{K18} made their fluid and diffusion systems on a closed manifold or an evolving closed surface under their assumptions. On the other hand, this paper makes a mathematical model for compressible fluid flow on an evolving surface with a boundary. Especially, we deal with the boundary conditions in co-normal direction for our model.

Let us first introduce basic notations. Let $x = { }^t ( x_1 , x_2, x_3) \in \mathbb{R}^3$, $\xi = { }^t ( \xi_1 , \xi_2 , \xi_3 ) \in \mathbb{R}^3$, $X = { }^t (X_1 , X_2) \in \mathbb{R}^2$ be the \emph{spatial variables}, and $t, \tau \geq 0$ be the \emph{time variables}. Let $T \in (0, \infty]$, and let $\Gamma (t) (= \{ \Gamma (t) \}_{0 \leq t < T}) $ be an \emph{evolving surface with a smooth boundary}. The symbol $n = n (x,t) = { }^t (n_1,n_2 , n_3)$ denotes the \emph{unit outer normal vector} at $x \in \overline{\Gamma (t)}$, and $\nu = \nu ( x, t) = { }^t (\nu_1 , \nu_2 , \nu_3 )$ denotes the \emph{unit outer co-normal vector} at $x \in \partial \Gamma (t)$, where $\partial \Gamma (t)$ denotes the boundary of the surface $\Gamma (t)$. The notations $\rho = \rho ( x, t )$, $\sigma = \sigma (x,t)$, $\theta = \theta (x , t)$, $e = e (x,t)$, and $C = C (x,t)$ represent the \emph{density}, the \emph{total pressure}, the \emph{temperature}, the \emph{internal energy} of the fluid on $\Gamma (t)$, and the \emph{concentration of amount of substance} in the fluid on $\Gamma (t)$, respectively. Note that the total pressure $\sigma$ includes both surface pressure and surface tension in general. The symbol $F = F (x , t) = { }^t (F_1 , F_2 , F_3)$ denotes an \emph{exterior force}. The notation $u = u (x , t ) = { }^t (u_1, u_2 , u_3 )$ denotes the \emph{relative fluid velocity} of a fluid particle at a point $x = { }^t (x_1 , x_2 , x_3 )$ of the evolving surface $\Gamma (t)$, and $w = w (x , t) = { }^t (w_1 , w_2 ,w_3)$ means the \emph{motion velocity} of the evolving surface $\Gamma (t)$. We often call $u$ \emph{surface flow} and $w$ the \emph{speed} of the evolving surface $\Gamma (t)$ (see \cite[Fig.1]{K18} for $u$ and $w$). By introducing the surface flow $u$ and the motion velocity $w$, then there is no exchange of particles between the surface and the environment. The velocity
\begin{equation*}
v = v (x , t) : = { }^t ( v_1 , v_2 , v_3 ) = u + w
\end{equation*}
is defined as the \emph{total velocity} of a fluid particle at a point $x$ of $\Gamma (t)$. In this paper we focus on the total velocity $v$, and assume that $\rho$, $\sigma$, $\theta$, $e$, $C$, $F$ and $v$ are smooth functions in $\mathbb{R}^4$. Set
\begin{equation*}
\Gamma_T = \bigcup_{0< t <T}\{ \Gamma (t) \times \{ t \} \}, { \ } \partial \Gamma_T = \bigcup_{0< t <T}\{ \partial \Gamma (t) \times \{ t \} \}.
\end{equation*}

This paper has four purposes. The first one is to make an abstract model for compressible fluid flow on the evolving surface $\Gamma (t)$ from an energetic point of view. More precisely, we apply both an energetic variational approach and the first law of thermodynamics to derive the following \emph{generalized compressible fluid system} on $\Gamma_T$:
\begin{equation}\label{eq11}
\begin{cases}
D_t \rho + ({\rm{div}}_\Gamma v) \rho = 0 & \text{ on } \Gamma_T,\\
\rho D_t v = {\rm{div}}_\Gamma S_{\Gamma} (v, \sigma ) + \rho F& \text{ on } \Gamma_T,\\
\rho D_t e + ({\rm{div}}_\Gamma v) \sigma = {\rm{div}}_\Gamma q_\theta + \tilde{e}_D& \text{ on } \Gamma_T,\\
D_t C + ({\rm{div}}_\Gamma v) C = {\rm{div}}_\Gamma q_C & \text{ on } \Gamma_T ,
\end{cases}
\end{equation}
where $D_t f = \partial_t f + ( v \cdot \nabla ) f$,
\begin{align*}
S_\Gamma (v , \sigma ) & = e_1' (| D_\Gamma (v) |^2 ) D_\Gamma (v) + e_2' ( | {\rm{div}}_\Gamma v |^2 ) ({\rm{div}}_\Gamma v ) P_\Gamma  - \sigma P_\Gamma ,\\
q_\theta & =  e_3'( | {\rm{grad}}_\Gamma \theta  |^2 ) {\rm{grad}}_\Gamma \theta,\\
\tilde{e}_D & = e_1' (| D_\Gamma (v) |^2 ) |D_\Gamma (v)|^2 + e_2' (|{\rm{div}}_\Gamma v|^2 ) | {\rm{div}}_\Gamma v |^2,\\
q_C & =  e_4'( | {\rm{grad}}_\Gamma C  |^2 ) {\rm{grad}}_\Gamma C.
\end{align*}
Here $| D_\Gamma (v) |^2 = D_\Gamma (v):D_\Gamma (v)$, $D_\Gamma (v) = P_\Gamma D (v) P_\Gamma$, $P_\Gamma = I_{3 \times 3} - n \otimes n$, $D (v) = \{ (\nabla v ) + { }^t ( \nabla v) \}/2$, $\nabla = { }^t (\partial_1 , \partial_2 , \partial_3)$, $\partial_j = \partial/{\partial x_j}$, $\partial_t = \partial/{\partial t}$, $\otimes $ denotes the tensor product, $e_{\mathfrak{j}}$ is a $C^1$-function, and $e_\mathfrak{j}' = e_\mathfrak{j}'( r ) = {d e_\mathfrak{j}}/{d r}$ $(\mathfrak{j}=1,2,3,4)$. See Section \ref{sect2} for the surface divergence ${\rm{div}}_\Gamma$ and surface gradient ${\rm{grad}}_\Gamma$. In this paper, we call $D_\Gamma (v)$, $S_\Gamma (v , \sigma)$, $\tilde{e}_D$, and $P_\Gamma$, the \emph{generalized surface stretching tensor}, the \emph{generalized surface stress tensor}, the \emph{generalized density for the energy dissipation due to the viscosities}, and an \emph{orthogonal projection to a tangent space}. We also call $q_\theta$ and $q_C$ the \emph{generalized heat flux} and the \emph{generalized surface flux}, respectively. We often call $S_\Gamma (v , \sigma)$ the \emph{surface stress tensor determined by the Boussinesq-Scriven law} if $e_1 (r ) = \mu_1 r$ and $e_2 ( r ) = \mu_2 r$ for some $\mu_1, \mu_2 \in \mathbb{R}$. Note also that
\begin{align*}
{\rm{div}}_\Gamma S_\Gamma (v , \sigma ) &= {\rm{div}}_\Gamma \{ \mu_1 D_\Gamma (v) + \mu_2 ({\rm{div}}_\Gamma v) P_\Gamma  \} - {\rm{grad}}_\Gamma \sigma - \sigma H_\Gamma n,\\ 
{\rm{div}}_\Gamma q_\theta &= \mu_3 \Delta_\Gamma \theta,\\
{\rm{div}}_\Gamma q_C & = \mu_4 \Delta_\Gamma C
\end{align*}
if $e_{\mathfrak{j}} ( r ) = \mu_{\mathfrak{j}} r$ for some $\mu_{\mathfrak{j}} \in \mathbb{R}$, where $H_\Gamma$ is the mean curvature in the direction $n$, and $\Delta_\Gamma$ is the Laplace-Beltrami operator. See Section \ref{sect2} for details.

The second one is to study the conservative form, the enthalpy, the entropy, the Helmholtz free energy, and the conservation laws of the system \eqref{eq11}. In fact, we can write the system \eqref{eq11} as follows:
\begin{equation}\label{eq12}
\begin{cases}
D_t^N \rho + {\rm{div}}_\Gamma (\rho v) =0& \text{ on } \Gamma_T,\\
D_t^N (\rho v) + {\rm{div}}_\Gamma \{ \rho v \otimes v - S_\Gamma (v,\sigma ) \} = \rho F& \text{ on } \Gamma_T,\\
D_t^N E_S + {\rm{div}}_\Gamma \{ E_S v - q_\theta - S_\Gamma (v, \sigma ) v \}  = \rho F \cdot v& \text{ on } \Gamma_T,\\
D_t^N C  + {\rm{div}}_\Gamma \{ C v - q_C \} = 0 & \text{ on } \Gamma_T.
\end{cases}
\end{equation}
Here $E_S$ is the \emph{total energy} defined by $E_S = \rho |v|^2/2 + \rho e$, and $D_t^N$ is the \emph{time derivative with the normal derivative} defined by
\begin{equation*}
D_t^N f = \partial_t f + ( v \cdot n ) (n \cdot \nabla ) f.
\end{equation*}
Under the assumptions that $h = e + \sigma / \rho $, $D_t e = \theta D_t s - \sigma D_t (1/\rho)$, $e_F = e - \theta s$, and $\rho , \theta >0$, the \emph{enthalpy} $h = h (x,t)$, the \emph{entropy} $s = s (x,t)$, and the \emph{Helmholtz free energy} $e_F = e_F (x,t)$ satisfy
\begin{align}
& D_t^N (\rho h) + {\rm{div}}_\Gamma \{ \rho h v - q_\theta \} = \tilde{e}_D + D_t \sigma & \text{ on }\Gamma_T, \label{eq13}\\
& D_t^N (\rho s ) + {\rm{div}}_\Gamma \left\{ \rho s v - \frac{q_\theta}{\theta} \right\} = \frac{\tilde{e}_D}{\theta} + \frac{e_3'(| {\rm{grad}}_\Gamma \theta |^2 ) |{\rm{grad}}_\Gamma \theta |^2 }{\theta^2} & \text{ on }\Gamma_T, \label{eq14}
\end{align}
and
\begin{equation*}
\rho D_t e_F + s \rho D_t \theta - S_\Gamma (v , \sigma ) : D_\Gamma (v) = - \tilde{e}_D \text{ on }\Gamma_T. 
\end{equation*}
We call \eqref{eq14} the \emph{generalized Clausius-Duhem inequality} if $e'_1,e'_2,e'_3$ are non-negative functions. Moreover, if system \eqref{eq11} satisfies the following boundary conditions in co-normal direction:
\begin{align}
S_\Gamma (v , \sigma ) \nu |_{\partial \Gamma_T} & = { }^t (0,0,0),\label{eq15}\\
\frac{\partial \theta }{\partial \nu} \bigg|_{\partial \Gamma_T} & =0 ,\label{eq16}\\
\frac{\partial C }{\partial \nu} \bigg|_{\partial \Gamma_T} & = 0 ,\label{eq17}
\end{align}
then any solution to system \eqref{eq11} satisfies that for $t_1 < t_2$,
\begin{align}
\int_{\Gamma (t_2)} \rho { \ }d \mathcal{H}^2_x & = \int_{\Gamma (t_1)} \rho { \ }d \mathcal{H}^2_x,\label{eq18}\\
\int_{\Gamma (t_2)} \rho v { \ }d \mathcal{H}^2_x & = \int_{\Gamma (t_1)} \rho v{ \ }d \mathcal{H}^2_x + \int_{t_1}^{t_2} \int_{\Gamma ( \tau )} \rho F { \ } d \mathcal{H}^2_x d \tau,\label{eq19}\\
\int_{\Gamma (t_2)}x \times \rho v { \ }d \mathcal{H}^2_x & = \int_{\Gamma (t_1)} x \times \rho v{ \ }d \mathcal{H}^2_x + \int_{t_1}^{t_2} \int_{\Gamma ( \tau )} x \times \rho F { \ } d \mathcal{H}^2_x d \tau, \label{Eq110}\\
\int_{\Gamma (t_2)} E_S { \ }d \mathcal{H}^2_x & = \int_{\Gamma (t_1)} E_S { \ }d \mathcal{H}^2_x + \int_{t_1}^{t_2} \int_{\Gamma (\tau)} \rho F \cdot v { \ }d \mathcal{H}^2_x d \tau ,\label{Eq111}\\
\int_{\Gamma (t_2)}  C { \ }d \mathcal{H}^2_x & = \int_{\Gamma (t_1)} C{ \ }d \mathcal{H}^2_x. \label{Eq112}
\end{align}
Here $d \mathcal{H}^k_x$ denotes the \emph{k-dimensional Hausdorff measure} and
\begin{equation*}
\frac{\partial f }{\partial \nu} := ( \nu \cdot \nabla_\Gamma ) f,
\end{equation*}
where $\nabla_\Gamma = {\rm{grad}}_\Gamma$. We often call \eqref{eq18}, \eqref{eq19}, \eqref{Eq110}, and \eqref{Eq111}, the \emph{law of conservation of mass}, the \emph{law of conservation of momentum}, the \emph{law of conservation of angular momentum}, and  the \emph{law of conservation of total energy}, respectively. See Theorem \ref{Thm210} for details.

The third one is to make an abstract model for the tangential compressible fluid flow on the evolving surface $\Gamma (t)$ from an energetic point of view. More precisely, we apply both an energetic variational approach and the first law of thermodynamics to derive the following \emph{generalized tangential compressible fluid system} on $\Gamma_T$:
\begin{equation}\label{Eq113}
\begin{cases}
v \cdot n =0 & \text{ on } \Gamma_T,\\
D^\Gamma_t \rho + ({\rm{div}}_\Gamma v) \rho = 0 & \text{ on } \Gamma_T,\\
P_\Gamma \rho D^\Gamma_t v = P_\Gamma {\rm{div}}_\Gamma S_{\Gamma} (v, \sigma ) + P_\Gamma \rho F& \text{ on } \Gamma_T,\\
\rho D^\Gamma_t e + ({\rm{div}}_\Gamma v) \sigma = {\rm{div}}_\Gamma q_\theta + \tilde{e}_D& \text{ on } \Gamma_T,\\
D^\Gamma_t C + ({\rm{div}}_\Gamma v) C = {\rm{div}}_\Gamma q_C & \text{ on } \Gamma_T,
\end{cases}
\end{equation}
where $D^\Gamma_t f := \partial_t f + ( v \cdot \nabla_\Gamma ) f $. Note that $D_t f = D^\Gamma_t f $ if $v \cdot n = 0$. Note that if system \eqref{Eq113} satisfies ether \eqref{eq15} or
\begin{equation}\label{Eq114}
v |_{\partial \Gamma_T} = { }^t (0,0,0),
\end{equation}
then any solution to system \eqref{Eq113} satisfies that for $t_1 < t_2$,
\begin{multline}\label{Eq115}
\int_{\Gamma (t_2)} \frac{1}{2} \rho |v|^2 { \ }d \mathcal{H}^2_x + \int_{t_1}^{t_2} \int_{\Gamma (\tau )} \tilde{e}_D { \ }d \mathcal{H}^2_x d \tau \\
= \int_{\Gamma (t_1)} \frac{1}{2} \rho |v|^2 { \ }d \mathcal{H}^2_x + \int_{t_1}^{t_2} \int_{\Gamma (\tau )} \{ ({\rm{div}}_\Gamma v ) \sigma + \rho F \cdot v \} { \ }d \mathcal{H}^2_x d \tau .
\end{multline}
This is the energy law of system \eqref{Eq113}. Note also that system \eqref{eq11} satisfies \eqref{Eq115} if either \eqref{eq15} or \eqref{Eq114} holds. See Theorem \ref{Thm211} for details.

The fourth one is to derive the following \emph{barotropic compressible fluid system} on $\Gamma_T$:
\begin{equation}\label{Eq116}
\begin{cases}
D_t \rho + ({\rm{div}}_\Gamma v) \rho = 0 & \text{ on }\Gamma_T,\\
\rho D_t v + {\rm{grad}}_\Gamma \mathfrak{p} + \mathfrak{p} H_\Gamma n =0 & \text{ on }\Gamma_T,\\
\mathfrak{p} = \mathfrak{p} ( \rho ) = \rho p' (\rho ) - p ( \rho ),
\end{cases} 
\end{equation}
and the \emph{tangential barotropic compressible fluid system} on $\Gamma_T$:
\begin{equation}\label{Eq117}
\begin{cases}
v \cdot n =0 & \text{ on }\Gamma_T,\\ 
D^\Gamma_t \rho + ({\rm{div}}_\Gamma v) \rho = 0 & \text{ on }\Gamma_T,\\
P_\Gamma \rho D^\Gamma_t v + {\rm{grad}}_\Gamma \mathfrak{p} =0 & \text{ on }\Gamma_T,\\
\mathfrak{p} = \mathfrak{p} ( \rho ) = \rho p' (\rho ) - p ( \rho ),
\end{cases} 
\end{equation}
where $p (\cdot)$ is a $C^1$ function. Note that if system \eqref{Eq116} satisfies \eqref{Eq114} then any solution to system \eqref{Eq116} satisfies that for $t_1 < t_2$,
\begin{equation}\label{Eq118}
\int_{\Gamma (t_2)} \frac{1}{2} \rho |v|^2 { \ }d \mathcal{H}^2_x = \int_{\Gamma (t_1)} \frac{1}{2} \rho |v|^2 { \ }d \mathcal{H}^2_x + \int_{t_1}^{t_2} \int_{\Gamma (\tau )}  ({\rm{div}}_\Gamma v ) \mathfrak{p} { \ }d \mathcal{H}^2_x d \tau .
\end{equation}
Note also that any solution to system \eqref{Eq117} satisfies \eqref{Eq118} if system \eqref{Eq117} satisfies \eqref{Eq114}.

\begin{remark}\label{rem11}
If we assume that $\sigma = \sigma (\rho , e)$ and $e = e (\rho , \theta )$, then in general system \eqref{eq11} with \eqref{eq15}-\eqref{eq17} is an overdetermined system for its initial value problem when the motion of $\Gamma (t)$ is given and we consider $(\rho, v , \theta , C)$ as unknown functions, while system \eqref{Eq113} with \eqref{Eq114}, \eqref{eq16}, \eqref{eq17} is not an overdetermined system for its initial value problem when the motion of $\Gamma (t)$ is given. Because the second expression of system \eqref{eq11} has six equations including the tangential and normal parts of the total velocity. For the same reason, system \eqref{Eq116} with \eqref{Eq114} is an overdetermined system for its initial value problem when the motion of $\Gamma (t)$ is given, and system \eqref{Eq117} with \eqref{Eq114} is not an overdetermined system for its initial value problem when the motion of $\Gamma (t)$ is given. When $v\cdot n =0$, we have to assume that $w \cdot n = 0$ since $u \cdot n =0$ in general.
\end{remark}

Let us explain our strategy for deriving four systems \eqref{eq11}, \eqref{Eq113}, \eqref{Eq116}, and \eqref{Eq117}. We first set the following \emph{generalized energy densities for compressible fluid on the evolving surface $\Gamma (t)$}.
\begin{assumption}[Energy densities for compressible fluid]\label{ass12}
\begin{multline*}
e_K = \frac{1}{2} \rho | v |^2,{ \ }e_P = p( \rho ),{ \ }e_D = \frac{1}{2} e_1 ( | D_\Gamma (v) |^2 )  + \frac{1}{2} e_2 ( | {\rm{div}}_\Gamma v |^2),\\
e_W = ({\rm{div}}_\Gamma v ) \sigma + \rho F \cdot v ,{ \ }e_{TD} = \frac{1}{2} e_3 ( | {\rm{grad}}_\Gamma \theta |^2), { \ }e_{GD} = \frac{1}{2} e_4 ( | {\rm{grad}}_\Gamma C |^2).
\end{multline*}
Here $p$, $e_1$, $e_2$, $e_3$, and $e_4$ are $C^1$-functions.
\end{assumption}
\noindent We call $e_K$ the \emph{kinetic energy}, $e_P$ the \emph{pressure potential}, $e_D$ the \emph{generalized energy density for the energy dissipation due to the viscosities}, $e_{W}$ the \emph{power density for the work done by the pressure and exterior force}, $e_{TD}$ the \emph{generalized energy density for the energy dissipation due to thermal diffusion}, and $e_{GD}$ the \emph{generalized energy density for the energy dissipation due to general diffusion}. Note that $e_D \neq \tilde{e}_D$ in general, where $\tilde{e}_D$ is the generalized density for the energy dissipation due to the viscosities defined by $\tilde{e}_D = e_1' (| D_\Gamma (v) |^2 ) |D_\Gamma (v)|^2 + e_2' (|{\rm{div}}_\Gamma v|^2 ) | {\rm{div}}_\Gamma v |^2$.

Secondly, we consider a mathematical validity of the generalized energy densities for compressible fluid on the evolving surface $\Gamma (t)$. We apply a flow map and the Riemannian metric induced by the flow map to show the invariance of our energy densities (see Section \ref{sect3}).

Thirdly, we derive several forces from a variation of some energies based on the generalized energy densities. We calculate a variation of our dissipation energies with respect to the velocity to have the viscous and diffusion terms of our fluid systems, and derive the nonlinear and pressure terms of our fluid systems to consider a variation of the action integral with respect to the flow map (see Section \ref{sect4}).

Finally, we apply both an energetic variational approach and the first law of thermodynamics to make mathematical models for compressible fluid flow on the evolving surface $\Gamma (t)$ with a boundary, and investigate the conservation and energy laws of our models (see Section \ref{sect5}). 

Let us state two key ideas of this paper. The first one is to generalize the energy densities for fluid. Using the energy densities for Newtonian fluid on a surface, Koba-Liu-Giga \cite{KLG17} and Koba \cite{K18} derived their fluid systems on an evolving closed surface. This paper generalizes the energy density for fluid on a surface to deal with non-Newtonian fluid and power law fluid. Note that our models \eqref{eq11} and \eqref{Eq113} includes their models in \cite{K18} and that our energy densities have a mathematical validity. Note also that Taylor's model \cite{Tay92} is a Newtonian fluid system on a closed surface. The second one is to make use of an energetic variational approach. An energetic variational approach is a method for deriving PDEs by applying the forces derived from a variation of energies. Gyarmati \cite{Gya70} used an energetic variational approach, which had been studied by Strutt \cite{Str73} and Onsager (\cite{Ons31a}, \cite{Ons31b}), to make several models for fluid dynamics in domains. Hyon-Kwak-Liu \cite{HKL10} applied an energetic variational approach to study complex fluid in domains. Koba-Sato \cite{KS17} used their energetic variational approach to derive their non-Newtonian fluid systems in domains. Koba-Liu-Giga \cite{KLG17} and Koba \cite{K18} improved the previous energetic variational approaches to derive their fluid systems on an evolving closed surface. This paper improves and polishes up their energetic variational approach in \cite{KLG17} and \cite{K18} to make a mathematical model for compressible fluid flow on an evolving surface with a boundary. 

Let us explain some derivations of incompressible fluid systems on a closed manifold. Arnol'd \cite{Arn66}, \cite{Arn97} applied the Lie group of diffeomorphisms to derive an inviscid incompressible fluid system on a manifold. See also Ebin-Marsden \cite{EM70}. Taylor \cite{Tay92} introduced their surface stretching tensor to make their viscous incompressible fluid system on a manifold. Mitsumatsu-Yano \cite{MY02} used their energetic variational approach to derive a viscous incompressible fluid system on a manifold. Arnaudon-Cruzeiro \cite{AC12} applied their stochastic variational approach to derive a viscous incompressible fluid system on a manifold. This paper considers compressible fluid flow on an evolving surface with a boundary.

Next we state some derivations of diffusion and fluid systems on an evolving closed surface. Dziuk-Elliott \cite{DE07} applied the surface transport theorem and their surface flux to make several diffusion systems on an evolving closed surface. Koba-Liu-Giga \cite{KLG17} applied an energetic variational approach and the Helmholtz-Weyl decomposition on a closed surface to derive their incompressible fluid systems on an evolving closed surface. Koba \cite{K18} studied compressible fluid flow on an evolving closed surface and derived the surface stress tensor determined by the Boussinesq-Scriven law from an energetic point of view. This paper derives compressible fluid systems on an evolving surface with a boundary, and investigates the natural boundary conditions for our fluid systems.

Now we state the history of the surface stress tensor determined by the Boussinesq-Scriven law. Boussinesq \cite{Bou13} first considered the existence of surface flow. Scriven \cite{Scr60} introduced their surface stress tensor to apply it to arbitrary surfaces. Slattery \cite{Sla64} studied some properties of the surface stress tensor determined by the Boussinesq-Scriven law. Bothe and Pr\"{u}ss \cite{BP12} used the Boussinesq-Scriven law to make a two-phase flow system with surface viscosity and surface tension. Koba \cite{K18} gave a mathematical validity of the Boussinesq-Scriven law from an energetic point of view. This paper derives the generalized surface stress tensor including the Boussinesq-Scriven law from an energetic point of view.

This paper provides a mathematical modeling of compressible fluid flow on an evolving with a boundary. We refer the reader to Serrin \cite{Ser59} for mathematical derivations of fluid systems in domains, Gatignol-Prud'homme \cite{GP01} and Slattery-Sagis-Oh \cite{SSO07} for two-phase flow systems with interfacial phenomena, and Gyarmati \cite{Gya70} and Gurtin-Fried-Anand \cite{GFA10} for the theory of thermodynamics.
The outline of this paper is as follows: In Section \ref{sect2} we state the definition of an evolving surface with a boundary and the main results of this paper. In Section \ref{sect3} we study the representation of our energy densities to show a mathematical validity of the energy densities. In Section \ref{sect4} we calculate variations of our dissipation energies and action integral to derive several forces from our energies. In Section \ref{sect5} we apply an energetic variational approach and the first law of thermodynamics to make mathematical models for compressible fluid flow on an evolving surface with a boundary, and investigate the conservation and energy laws of our fluid systems.

\section{Main Results}\label{sect2}

Let us first introduce a bounded domain with a boundary, a surface with a boundary, and an evolving surface with a boundary. Then we define fundamental notations. Finally, we state the main results of this paper.
\begin{definition}[Bounded domain with a boundary]\label{def21}{ \ }
Let $U \subset \mathbb{R}^2$ be a bounded domain. We call $U$ a \emph{bounded domain with a smooth boundary} if the boundary $\partial U$ of $U$ can be written as
\begin{equation*}
\partial U = \partial U^1 \cup \partial U^2 \cup \cdots \cup \partial U^M \text{ for some } M \in \mathbb{N}.
\end{equation*}
Here
\begin{equation*}
\partial U^m = \{ X = { }^t ( X_1 , X_2 ) \in \mathbb{R}^2; { \ } X_1 = p_m ( r ), { \ }X_2 = q_m ( r ),{ \ } r \in [a_m , b_m ] \},
\end{equation*}
where $m \in \{1,2, \cdots, M\}$, $a_m , b_m \in \mathbb{R}$, $p_m, q_m \in C^2([ a_m , b_m ])$, satisfying the following three properties:\\ 
$(\mathrm{i})$ $(p_{m} (b_m ) , q_{m} ( b_m)) = (p_{m + 1} (a_{m+1} ) , q_{m+1} (a_{m+1}))$, $(p'_{m} (b_m ) , q'_{m} ( b_m)) =$\\ 
$ (p'_{m + 1} (a_{m+1} ) , q'_{m+1} (a_{m+1}))$, $(p^{(2)}_{m} (b_m ) , q^{(2)}_{m} ( b_m)) = (p^{(2)}_{m + 1} (a_{m+1} ) , q^{(2)}_{m+1} (a_{m+1}))$ for each $m \in \{1 ,2 , \cdots, M - 1 \}$.\\
$(\mathrm{ii})$ $(p_{M} (b_M ) , q_{M} (b_M)) = (p_1 (a_1 ) , q_1 (a_1))$, $(p'_{M} (b_M ) , q'_{M} (b_M)) = (p'_1 (a_1 ) , q'_1 (a_1))$, $(p^{(2)}_{M} (b_M ) , q^{(2)}_{M} (b_M)) = (p^{(2)}_1 (a_1 ) , q^{(2)}_1 (a_1))$.\\
$(\mathrm{iii})$ For each $m \in \{1 ,2 , \cdots, M \}$
\begin{equation*}
\left( \frac{d p_m }{d r} ( r ) \right)^2 + \left( \frac{d q_m}{d r} ( r ) \right)^2 > 0 \text{ for all } r \in [ a_m , b_m ].
\end{equation*}
Here $f' = f'(r) = {df}/{dr}$ and $f^{(2)} = f^{(2)}(r) ={d^2f}/{d r^2}$. 
\end{definition}
\begin{definition}[Surface with a boundary]\label{def22}
Let $\Gamma_0 \subset \mathbb{R}^3$ be a set. We call $\Gamma_0$ a \emph{surface with a smooth boundary} if there are bounded domain $U$ with a smooth boundary and $\Phi = { }^t ( \Phi_1 , \Phi_2 , \Phi_3) \in [ C^2 ( \overline{U})]^3$ satisfying the properties as in Definition \ref{def21} and the following four properties:\\
$(\mathrm{i})$ The set $\Gamma_0$ can be written as
\begin{equation*}
\Gamma_0 = \{ \xi = { }^t (\xi_1, \xi_2, \xi_3) \in \mathbb{R}^3; { \ }\xi = \Phi (X) , { \ }X \in U \}.
\end{equation*}
$(\mathrm{ii})$ The map
\begin{equation*}
\Phi : U \to \Gamma_0 \text{ is bijective}.
\end{equation*}
$(\mathrm{iii})$ The boundary $\partial \Gamma_0$ of $\Gamma_0$ can be written as
\begin{equation*}
\partial \Gamma_0 = \partial \Gamma_0^1 \cup \partial \Gamma_0^2 \cup \cdots \cup \partial \Gamma_0^M,
\end{equation*}
where
\begin{equation*}
\partial \Gamma_0^m = \{ \xi \in \mathbb{R}^3; { \ } \xi = \Phi (X), { \ }X \in \partial U^m \}.
\end{equation*}
$(\mathrm{iv})$ There is $\lambda_1 >0$ such that for every $X \in \overline{U}$,
\begin{multline*}
\left| \frac{\partial \Phi}{\partial X_1} \times \frac{\partial \Phi}{\partial X_2} \right|^2 = \left( \frac{\partial \Phi_2}{\partial X_1} \frac{\partial \Phi_3}{\partial X_2} - \frac{\partial \Phi_2}{\partial X_2} \frac{\partial \Phi_3}{\partial X_1} \right)^2 + \left( \frac{\partial \Phi_1}{\partial X_1} \frac{\partial \Phi_3}{\partial X_2} - \frac{\partial \Phi_1}{\partial X_2} \frac{\partial \Phi_3}{\partial X_1} \right)^2\\
 + \left( \frac{\partial \Phi_1}{\partial X_1} \frac{\partial \Phi_2}{\partial X_2} - \frac{\partial \Phi_1}{\partial X_2} \frac{\partial \Phi_2}{\partial X_1} \right)^2 \geq \lambda_1.
\end{multline*}
Here
\begin{equation*}
C^2 ( \overline{U}) = \{ f : \overline{U} \to \mathbb{R}; { \ }f = \mathcal{F}|_{\overline{U} } \text{ for some } \mathcal{F} \in C^2 ( \mathbb{R}^2)\}.
\end{equation*}
\end{definition}

\begin{definition}[Evolving surface with a boundary and flow maps]\label{def23}{ \ }\\
Let $\Gamma_0 \subset \mathbb{R}^3$ be a surface with a smooth boundary, and let $U \subset \mathbb{R}^2$ and $\Phi = { }^t ( \Phi_1 , \Phi_2 , \Phi_3 ) \in [ C^2 ( \overline{U})]^3$ be the bounded domain with a smooth boundary and the functions satisfying the properties as in Definitions \ref{def21} and \ref{def22}. For $0 \leq t < T$, let $\Gamma (t) \subset \mathbb{R}^3$ be a surface with a smooth boundary. We call $\Gamma (t) (= \{ \Gamma (t) \}_{0 \leq t <T})$ an \emph{evolving surface with a smooth boundary} if there is $\widetilde{x} = { }^t ( \widetilde{x}_1 , \widetilde{x}_2 , \widetilde{x}_3) \in [ C^3 ( \overline{\Gamma_0} \times [0,T ))]^3$ satisfying the following six properties:\\
$(\mathrm{i})$ For every $\xi \in \overline{\Gamma_0}$
\begin{equation*}
\widetilde{x} ( \xi , 0) = \xi .
\end{equation*}
$(\mathrm{ii})$ For every $0 < t < T$, $\Gamma (t)$ can be written as
\begin{equation*}
\Gamma (t) = \{ x = { }^t (x_1, x_2,x_3) \in \mathbb{R}^3; { \ }x = \widetilde{x}( \xi , t) ,{ \ } \xi \in \Gamma_0 \}.
\end{equation*}
$(\mathrm{iii})$ For each $0 < t < T$,
\begin{equation*}
\widetilde{x}( \cdot , t) : \Gamma_0 \to \Gamma (t) \text{ is bijective}.
\end{equation*}
$(\mathrm{iv})$ For every $0 < t < T$, the boundary $\partial \Gamma (t)$ of $\Gamma (t)$ can be written as
\begin{equation*}
\partial \Gamma (t) = \partial \Gamma^1 (t) \cup \partial \Gamma^2 (t) \cup \cdots \cup \partial \Gamma^M (t),
\end{equation*}
where
\begin{equation*}
\partial \Gamma^m (t) = \{ x \in \mathbb{R}^3; { \ } x = \widetilde{x} (\xi, t ),{ \ } \xi \in \partial \Gamma_0^m \}.
\end{equation*}
$(\mathrm{v})$ Set $\widehat{x} (X,t) = \widetilde{x} ( \Phi (X), t )$. There is $\lambda_2 >0$ such that for every $0 \leq t < T$ and $X \in \overline{U}$,
\begin{multline*}
\left| \frac{\partial \widehat{x}}{\partial X_1} \times \frac{\partial \widehat{x}}{\partial X_2} \right|^2 = \left( \frac{\partial \widehat{x}_2}{\partial X_1} \frac{\partial \widehat{x}_3}{\partial X_2} - \frac{\partial \widehat{x}_2}{\partial X_2} \frac{\partial \widehat{x}_3}{\partial X_1} \right)^2 + \left( \frac{\partial \widehat{x}_1}{\partial X_1} \frac{\partial \widehat{x}_3}{\partial X_2} - \frac{\partial \widehat{x}_1}{\partial X_2} \frac{\partial \widehat{x}_3}{\partial X_1} \right)^2\\
 + \left( \frac{\partial \widehat{x}_1}{\partial X_1} \frac{\partial \widehat{x}_2}{\partial X_2} - \frac{\partial \widehat{x}_1}{\partial X_2} \frac{\partial \widehat{x}_2}{\partial X_1} \right)^2 \geq \lambda_2.
\end{multline*}
$(\mathrm{vi})$ There exists $v = v (x,t) = { }^t (v_1, v_2 , v_3) \in C^2 ( \bigcup_{0 \leq t < T} \{ \overline{\Gamma (t)} \times \{ t \} \} )$ such that for every $0 \leq t < T$ and $\xi \in \overline{\Gamma_0}$
\begin{equation*}
\widetilde{x}_t (\xi , t ) = v ( \widetilde{x} ( \xi , t) , t ).
\end{equation*}
Here
\begin{multline*}
C^2 \bigg( \bigcup_{0 \leq t < T} \{ \overline{\Gamma (t)} \times \{ t \} \} \bigg) := \bigg\{ f : \bigcup_{0 \leq t < T} \{ \overline{\Gamma (t)} \times \{ t \} \} \to \mathbb{R};\\
 { \ }f = \mathcal{F}|_{\bigcup_{0 \leq t < T} \{ \overline{\Gamma (t)} \times \{ t \} \}} \text{ for some } \mathcal{F} \in C^2 ( \mathbb{R}^4) \bigg\}.
\end{multline*}
We call $\widetilde{x} = \widetilde{x} (\xi, t)$ a \emph{flow map on} $\Gamma (t)$, and $v = v ( x , t ) = v (\widetilde{x}(\xi, t), t)$ the \emph{velocity determined by the flow map} $\widetilde{x}$. We also call $\widehat{x} = \widehat{x} ( X , t)$ a \emph{flow map on} $\Gamma (t)$ since $\partial \widehat{x}/ \partial t (X,t) = v ( \widehat{x} (X , t ) , t)$. Note that $\widehat{x} (X,t) = \widetilde{x} ( \Phi (X) , t) = \widetilde{x} ( \xi , t) $.
\end{definition}
Let us explain the conventions used in this paper. We use the italic characters $i,j, k, \ell$ as $1,2,3$, and the Greek characters $\alpha , \beta, \zeta , \eta$ as $1,2$. Moreover, we often use the following Einstein summation convention: $c_j d_{ij} = \sum_{j = 1}^3 c_j d_{i j}$, $c_\alpha d^{\alpha \beta} = \sum_{\alpha = 1}^2 c_\alpha d^{\alpha \beta}$. The symbol $[\mathcal{M}]_{ \mathfrak{i} \mathfrak{j}}$ denotes the $( \mathfrak{i} , \mathfrak{j})$-th component of a matrix $\mathcal{M}$. For each $k \times k$-matrix $\mathcal{M}$, $ | \mathcal{M} |^2 := \mathcal{M} : \mathcal{M} = \sum_{\mathfrak{i}, \mathfrak{j} =1}^k | [\mathcal{M}]_{\mathfrak{i} \mathfrak{j}}|^2$.

We now introduce function spaces and notations. Let $\Gamma (t) (= \{ \Gamma (t) \}_{0 \leq t < T})$ be an evolving surface with a smooth boundary, and let $\widetilde{x} = \widetilde{x} ( \xi , t) = { }^t ( \widetilde{x}_1, \widetilde{x}_2 , \widetilde{x}_3)$, $\widehat{x} = \widehat{x} (X , t) = { }^t ( \widehat{x}_1 , \widehat{x}_2 , \widehat{x}_3)$, and $v = v (x,t) = { }^t (v_1 , v_2 , v_3)$ be two flow maps on $\Gamma (t)$, and the velocity determined by the flow map $\widetilde{x}$. Throughout this paper we assume that $v$ is the total velocity. The symbol $n = n (x, t) = { }^t ( n_1 , n_2 , n_3)$ and $\nu = \nu (x,t) = { }^t ( \nu_1 , \nu_2 , \nu_3 )$ denote the unit outer normal vector at $ x \in \overline{\Gamma (t)}$ and the unit outer co-normal vector at $x \in \partial \Gamma (t)$, respectively. Notice that $\nu = \nu (x,t)$ exists for almost all $x \in \partial \Gamma (t)$ since $\partial \Gamma (t)$ is a smooth boundary. Set
\begin{align*}
\Gamma_T & = \bigcup_{0 <  t  < T}\{ \Gamma (t) \times \{ t \} \},{ \ }\partial \Gamma_T = \bigcup_{0 <  t  < T}\{ \partial \Gamma (t) \times \{ t \} \},\\
\overline{\Gamma_T} & = \bigcup_{0 \leq t  < T}\{ \overline{\Gamma (t)} \times \{ t \} \}.
\end{align*}
For each $k \in \mathbb{N} \cup \{ 0 , \infty \}$,
\begin{align*}
C^k ( \Gamma (t) ) & := \{ f : \Gamma (t) \to \mathbb{R}; { \ }f = \mathcal{F}|_{\Gamma (t)} \text{ for some } \mathcal{F} \in C^k ( \mathbb{R}^3) \},\\
C^k ( \overline{\Gamma (t)}) & := \{ f : \overline{\Gamma (t)} \to \mathbb{R}; { \ }f = \mathcal{F}|_{\overline{\Gamma (t)}} \text{ for some } \mathcal{F} \in C^k ( \mathbb{R}^3) \},\\
C_0^k (\Gamma (t)) & := \{ f \in C^k (\Gamma (t)); { \ }\text{supp} f \text{ does not intersect } \partial \Gamma (t) \},\\
C^k ( \overline{ \Gamma_T} ) & := \{ f: \overline{\Gamma_T} \to \mathbb{R}; { \ } f = \mathcal{F}|_{\overline{\Gamma_T}} \text{ for some } \mathcal{F} \in C^k ( \mathbb{R}^4) \},\\
C^k ( \partial \Gamma (t) ) & := \{ f:\partial \Gamma (t) \to \mathbb{R}; { \ } f = \mathcal{F}|_{\partial \Gamma (t)} \text{ for some } \mathcal{F} \in C^k ( \mathbb{R}^3) \},\\
C^k ( \partial \Gamma_T ) & := \{ f:\partial \Gamma_T \to \mathbb{R}; { \ } f = \mathcal{F}|_{\partial \Gamma_T} \text{ for some } \mathcal{F} \in C^k ( \mathbb{R}^4) \}.
\end{align*}
Fix $j \in \{ 1 , 2 , 3 \}$. For $f \in C^1 ( \overline{\Gamma (t)})$ and $\varphi = { }^t ( \varphi_1, \varphi_2 , \varphi_3) \in C^1 (\overline{\Gamma (t)})$,
\begin{align*}
\partial_j^\Gamma f & := (\delta_{i j} - n_i n_j ) \partial_i f = \sum_{i=1}^3 (\delta_{i j} - n_i n_j ) \partial_i f,\\ 
\nabla_\Gamma & :=  { }^t ( \partial_1^\Gamma , \partial_2^\Gamma , \partial_3^\Gamma ),\\
{\rm{grad}}_\Gamma f & := \nabla_\Gamma f = { }^t (\partial_1^\Gamma f, \partial_2^\Gamma f, \partial_3^\Gamma f) ,\\
{\rm{div}}_\Gamma \varphi & := \nabla_\Gamma \cdot \varphi = \partial_1^\Gamma \varphi_1 + \partial_2^\Gamma \varphi_2 + \partial_3^\Gamma \varphi_3,\\
\Delta_\Gamma f & := {\rm{div}}_\Gamma ({\rm{grad}}_\Gamma f ) = (\partial_1^\Gamma )^2 f +  (\partial_2^\Gamma )^2 f +  (\partial_3^\Gamma )^2 f.
\end{align*}
Here $\delta_{ij}$ denotes the Kronecker delta. Moreover, for $g \in C^1 ( \partial \Gamma (t) )$,
\begin{equation*}
\frac{\partial g}{\partial \nu} := ( \nu \cdot \nabla_\Gamma ) g . 
\end{equation*}
The symbol $P_\Gamma = P_\Gamma (x ,t)$ and $H_\Gamma = H_\Gamma (x , t )$ denote the \emph{orthogonal projection to a tangent space} and the \emph{mean curvature in the direction} $n$ defined by
\begin{align*}
P_\Gamma & = I_{3 \times 3} - n \otimes n,\\
{ \ }H_\Gamma & = - {\rm{div}}_\Gamma n .
\end{align*}
Note that $[P_\Gamma]_{i j} = \delta_{i j} - n_i n_j$, $P_\Gamma n = { }^t (0,0,0)$, $P_\Gamma^2 = P_\Gamma$, $H_\Gamma = - (\partial_1^\Gamma n_1 + \partial_2^\Gamma n_2 + \partial_3^\Gamma n_3)$ and that ${\rm{div}}_\Gamma (f P_\Gamma ) = {\rm{grad}}_\Gamma f + f H_\Gamma n$. Note also that $(\nu \cdot \nabla )f = ( \nu \cdot \nabla_\Gamma ) f $ and $H_\Gamma = - (\partial_1 n_1 + \partial_2 n_2 + \partial_3 n_3)$ since $\nu \cdot n = 0$ and $| n | = | \nu | = 1$.

Next we introduce the Riemannian metric determined by the flow map $\widehat{x} = \widehat{x} ( X , t )$. For every $0 \leq t < T$ and $X \in \overline{U}$,
\begin{multline*}
\mathfrak{g}_1 = \mathfrak{g}_1 (X,t) := \frac{\partial \widehat{x}}{\partial X_1},{ \ }\mathfrak{g}_2 = \mathfrak{g}_2 (X,t) := \frac{\partial \widehat{x}}{\partial X_2},{ \ }\mathfrak{g}_{\alpha \beta} := \mathfrak{g}_\alpha \cdot \mathfrak{g}_\beta = \frac{\partial \widehat{x}_i}{\partial X_\alpha} \frac{\partial \widehat{x}_i}{\partial X_\beta},\\
( \mathfrak{g}^{ \alpha \beta })_{2 \times 2} := ( \mathfrak{g}_{\alpha \beta })_{2 \times 2}^{-1} = \frac{1}{\mathfrak{g}_{11} \mathfrak{g}_{22} - \mathfrak{g}_{12} \mathfrak{g}_{21} } 
\begin{pmatrix}
\mathfrak{g}_{22} & - \mathfrak{g}_{21}\\
- \mathfrak{g}_{12} & \mathfrak{g}_{11}
\end{pmatrix},{ \ }\mathfrak{g}^\alpha := \mathfrak{g}^{\alpha \beta} \mathfrak{g}_\beta,\\
\acute{ \mathfrak{g} }_{\alpha } :=  \frac{d}{d t} \mathfrak{g}_\alpha, { \ }\acute{\mathfrak{g}}_{\alpha \beta} := \frac{d }{d t} \mathfrak{g}_{\alpha \beta }  , { \ }  \mathcal{G} = \mathcal{G}(X,t) := | \mathfrak{g}_1 \times \mathfrak{g}_2 |^2 = \mathfrak{g}_{11} \mathfrak{g}_{22} - \mathfrak{g}_{12} \mathfrak{g}_{21}. 
\end{multline*}
By Definition \ref{def23}, we see that for every $0 \leq t < T$ and $X \in \overline{U}$,
\begin{multline}\label{eq21}
\mathcal{G} = | \mathfrak{g}_1 \times \mathfrak{g}_2 |^2 =\mathfrak{g}_{11} \mathfrak{g}_{22} - \mathfrak{g}_{12} \mathfrak{g}_{21} = \left( \frac{\partial \widehat{x}_2}{\partial X_1} \frac{\partial \widehat{x}_3}{\partial X_2} - \frac{\partial \widehat{x}_2}{\partial X_2} \frac{\partial \widehat{x}_3}{\partial X_1} \right)^2\\
 + \left( \frac{\partial \widehat{x}_1}{\partial X_1} \frac{\partial \widehat{x}_3}{\partial X_2} - \frac{\partial \widehat{x}_1}{\partial X_2} \frac{\partial \widehat{x}_3}{\partial X_1} \right)^2   + \left( \frac{\partial \widehat{x}_1}{\partial X_1} \frac{\partial \widehat{x}_2}{\partial X_2} - \frac{\partial \widehat{x}_1}{\partial X_2} \frac{\partial \widehat{x}_2}{\partial X_1} \right)^2 \geq \lambda_2 > 0.
\end{multline}
See Jost \cite{Jos11} and Ciarlet \cite{Cia05} for differential geometry and the Riemannian manifold.

We now state the main results of this paper. Throughout this paper we assume that $\rho$, $v$, $\sigma$, $\theta$, $e$, $F$, and $C$ are $C^2$-functions on $\overline{\Gamma_T}$. Let us first study our energy densities (see Assumption \ref{ass12}).
\begin{theorem}[Representation of energy densities]\label{thm24}
Fix $t$. Then
\begin{align}
\int_{\Gamma (t)} {\rm{div}}_\Gamma v  { \ } d \mathcal{H}^2_x & = \int_U \frac{d \sqrt{\mathcal{G}}  }{d t}  { \ }d X,\label{eq22}\\
\int_{\Gamma (t)} ({\rm{div}}_\Gamma v ) \sigma { \ } d \mathcal{H}^2_x & = \int_U \left( \frac{1}{2} \acute{\mathfrak{g}}_{ \alpha \beta } \mathfrak{g}^{ \alpha \beta } \right) \sigma  \sqrt{ \mathcal{G} }  { \ }d X,\label{eq23}\\
\int_{\Gamma (t)} \frac{1}{2} e_1 ( |D_\Gamma (v) |^2) { \ }d \mathcal{H}^2_x & = \int_{U} \frac{1}{2} e_1 \left( \frac{1}{4} \acute{ \mathfrak{g} }_{\alpha \beta} \acute{ \mathfrak{g} }_{\zeta \eta} \mathfrak{g}^{\alpha \zeta} \mathfrak{g}^{\beta \eta } \right) \sqrt{ \mathcal{G} } { \ }d X,\label{eq24}\\
\int_{\Gamma (t)} \frac{1}{2} e_2 ( | {\rm{div}}_\Gamma v |^2) { \ }d \mathcal{H}^2_x & = \int_{U} \frac{1}{2} e_2 \left(\frac{1}{4} \acute{ \mathfrak{g} }_{\alpha \beta} \acute{ \mathfrak{g} }_{\zeta \eta} \mathfrak{g}^{\alpha \beta} \mathfrak{g}^{\zeta \eta } \right) \sqrt{ \mathcal{G}} { \ }d X,\label{eq25}\\
\int_{\Gamma (t)} \frac{1}{2} e_3 ( | {\rm{grad}}_\Gamma \theta |^2) { \ }d \mathcal{H}^2_x & = \int_{U} \frac{1}{2} e_3 \left( \mathfrak{g}^{\alpha \beta } \frac{\partial \widehat{\theta} }{\partial X_\alpha} \frac{\partial \widehat{\theta} }{\partial X_\beta} \right) \sqrt{ \mathcal{G} } { \ }d X,\label{eq26}\\
\int_{\Gamma (t)} \frac{1}{2} e_4 ( | {\rm{grad}}_\Gamma C |^2) { \ }d \mathcal{H}^2_x & = \int_{U} \frac{1}{2} e_4 \left( \mathfrak{g}^{\alpha \beta } \frac{\partial \widehat{C} }{\partial X_\alpha} \frac{\partial \widehat{C} }{\partial X_\beta} \right) \sqrt{ \mathcal{G} } { \ }d X.\label{eq27}
\end{align}
Here $\widehat{\theta} = \widehat{\theta} (X,t) := \theta ( \widehat{x} (X,t) , t)$ and  $\widehat{C} = \widehat{C} (X,t) := C ( \widehat{x} (X,t) , t)$. 
\end{theorem}

\begin{remark}\label{rem25}
From Koba-Sato \cite[Theorem 2.1]{KS17} and Theorem \ref{thm24} we see that our energy densities are candidates of the energy densities for compressible fluid flow on the evolving surface $\Gamma (t)$ since our energy densities are invariant throughout the Riemannian metric determined by flow maps.
\end{remark}

Next we introduce the surface divergence theorem to calculate a variation of our dissipation energies and work. 
\begin{proposition}[Divergence theorem on surface with boundary]\label{prop25}{ \ }\\
For every $\varphi = { }^t ( \varphi_1 , \varphi_2 , \varphi_3 ) \in C^1 ( \overline{ \Gamma (t)})$,
\begin{equation}
\int_{\Gamma (t)} {\rm{div}}_\Gamma \varphi { \ }d \mathcal{H}^2_x = - \int_{\Gamma (t)} H_\Gamma (n \cdot \varphi) { \ } d \mathcal{H}^2_x + \int_{\partial \Gamma (t)} \nu \cdot \varphi { \ } d \mathcal{H}^1_x.\label{eq28}
\end{equation}
Here $\nu = \nu (x,t) = { }^t (\nu_1 , \nu_2 , \nu_3 )$ denotes the unit outer co-normal vector to $\partial \Gamma (t)$, i.e., $\nu \cdot \nu = 1$ and $\nu \cdot n = 0$ on $\partial \Gamma (t)$.
\end{proposition}
\noindent Simon \cite{Sim83} proved Proposition \ref{prop25} in the case when the surface $\Gamma (t)$ have a smooth boundary, and Koba \cite{K19} proved Proposition \ref{prop25} in the case when the surface $\Gamma (t)$ have a piecewise Lipschitz-continuous boundary. From \cite{K19} we find that the unit outer co-normal vector $\nu$ is represented by
\begin{equation*}
\nu = \nu ( \widehat{x} (X,t) , t ) = \frac{n_1^U \mathfrak{g}_2 - n_2^U \mathfrak{g}_1}{ | n_1^U \mathfrak{g}_2 -n_2^U \mathfrak{g}_1| } \times \frac{\mathfrak{g}_1 \times \mathfrak{g}_2}{ | \mathfrak{g}_1 \times \mathfrak{g}_2 | } \text{ for a.e. } X \in \partial U.
\end{equation*}
Here $n^U = n^U (X) = { }^t (n_1^U , n_2^U)$ denotes the unit outer normal vector at $X \in \partial U$. From Proposition \ref{prop25} we have the following formula for the integration by parts:
\begin{equation}
\int_{\Gamma (t)} f (\partial_j^\Gamma g) { \ }d \mathcal{H}^2_x \\
= - \int_{\Gamma (t)} (\partial_j^\Gamma f + H_\Gamma n_j f) g { \ } d \mathcal{H}^2_x + \int_{\partial \Gamma (t)} (\nu_j f ) g { \ } d \mathcal{H}^1_x.\label{eq29}
\end{equation}

Now we define our dissipation energies and work, and state a variation of the dissipation energies. For every $- 1 < \varepsilon < 1$, $\varphi = { }^t ( \varphi_1 , \varphi_2 , \varphi_3 ) \in [C_0^1 ( \Gamma (t))]^3$, and $\psi \in C_0^1 ( \Gamma (t))$,
\begin{align*}
E_D [ v + \varepsilon \varphi ] (t) & := - \int_{\Gamma (t)}  \frac{1}{2} \{ e_1 ( | D_\Gamma (v + \varepsilon \varphi ) |^2 ) + e_2 ( | {\rm{div}}_\Gamma (v + \varepsilon \varphi ) |^2 ) \} { \ } d \mathcal{H}^2_x,\\
E_W [ v + \varepsilon \varphi ] (t) & := \int_{\Gamma (t)} [ \{ {\rm{div}}_\Gamma (v + \varepsilon \varphi ) \} \sigma + \rho F \cdot (v + \varepsilon \varphi ) ] { \ } d \mathcal{H}^2_x,\\
E_{TD} [ \theta + \varepsilon \psi ] (t) & := - \int_{\Gamma (t)}  \frac{1}{2} \{ e_3 ( | {\rm{grad}}_\Gamma (\theta + \varepsilon \psi ) |^2 ) \} { \ } d \mathcal{H}^2_x,\\
E_{GD} [ C + \varepsilon \psi ] (t) & := - \int_{\Gamma (t)}  \frac{1}{2} \{ e_4 ( | {\rm{grad}}_\Gamma ( C + \varepsilon \psi) |^2 ) \} { \ } d \mathcal{H}^2_x.
\end{align*}
\begin{theorem}[Variation of dissipation energies]\label{thm27}Fix $t$. Then\\
$(\mathrm{i})$ For every $\varphi = { }^t ( \varphi_1 , \varphi_2 , \varphi_3 ) \in [ C_0^1 (\Gamma (t) ) ]^3$,
\begin{multline*}
\frac{d}{d \varepsilon} \bigg|_{\varepsilon = 0} E_{D} [v + \varepsilon \varphi] (t)\\
 = \int_{\Gamma (t)} {\rm{div}}_\Gamma \{ e'_1 (| D_\Gamma (v) |^2 ) D_\Gamma (v) + e'_2 (| {\rm{div}}_\Gamma v |^2 )  ({\rm{div}}_\Gamma v) P_\Gamma  \} \cdot \varphi { \ }d \mathcal{H}^2_x.
\end{multline*}
Assume in addition that $\varphi \cdot n =0$ on $\Gamma (t)$. Then
\begin{multline*}
\frac{d}{d \varepsilon} \bigg|_{\varepsilon = 0} E_{D} [v + \varepsilon \varphi] (t)\\
 = \int_{\Gamma (t)} P_\Gamma {\rm{div}}_\Gamma \{ e'_1 (| D_\Gamma (v) |^2 ) D_\Gamma (v) + e'_2 (| {\rm{div}}_\Gamma v |^2 ) ({\rm{div}}_\Gamma v) P_\Gamma  \} \cdot \varphi { \ }d \mathcal{H}^2_x.
\end{multline*}
$(\mathrm{ii})$ For every $\varphi = { }^t ( \varphi_1 , \varphi_2 , \varphi_3 ) \in [ C_0^1 (\Gamma (t) ) ]^3$,
\begin{equation*}
\frac{d}{d \varepsilon} \bigg|_{\varepsilon = 0} E_W [v + \varepsilon \varphi] (t) = \int_{\Gamma (t)} \{ {\rm{div}}_\Gamma ( - \sigma P_\Gamma ) + \rho F \} \cdot \varphi { \ }d \mathcal{H}^2_x.
\end{equation*}
Assume in addition that $\varphi \cdot n =0$ on $\Gamma (t)$. Then
\begin{equation*}
\frac{d}{d \varepsilon} \bigg|_{\varepsilon = 0} E_W [v + \varepsilon \varphi] (t) = \int_{\Gamma (t)} \{ P_\Gamma {\rm{div}}_\Gamma ( - \sigma P_\Gamma ) + P_\Gamma \rho F \}\cdot \varphi { \ }d \mathcal{H}^2_x.
\end{equation*}
$(\mathrm{iii})$ For every $\psi \in C_0^1 ( \Gamma (t) )$,
\begin{equation*}
\frac{d}{d \varepsilon} \bigg|_{\varepsilon = 0} E_{TD} [\theta + \varepsilon \psi ] (t) = \int_{\Gamma (t)} {\rm{div}}_\Gamma \{ e'_3 (| {\rm{grad}}_\Gamma \theta |^2 ) {\rm{grad}}_\Gamma \theta \} \cdot \psi { \ }d \mathcal{H}^2_x.
\end{equation*}
$(\mathrm{iv})$ For every $\psi \in C_0^1 ( \Gamma (t) )$,
\begin{equation*}
\frac{d}{d \varepsilon} \bigg|_{\varepsilon = 0} E_{GD} [ C + \varepsilon \psi ] (t) = \int_{\Gamma (t)} {\rm{div}}_\Gamma \{ e'_4 (| {\rm{grad}}_\Gamma C |^2 ) {\rm{grad}}_\Gamma C \} \cdot \psi { \ }d \mathcal{H}^2_x.
\end{equation*}
\end{theorem}
\noindent From Theorem \ref{thm27} we have several forces derived from variations of our dissipation energies and work. In other words, this paper derives the viscous and diffusion terms of our fluid systems from a variation of the dissipation energies. See Sections \ref{sect4} and \ref{sect5} for details.

Now we consider a variation of the action integral determined by the kinetic energy with respect to the flow map. To this end, we introduce a variation $\widetilde{x}^\varepsilon$ of the flow map $\widetilde{x}$ and the velocity $v^\varepsilon$ determined by the flow map $\widetilde{x}^\varepsilon$. For $- 1 < \varepsilon < 1$, let $\Gamma^\varepsilon (t) (= \{ \Gamma^\varepsilon (t) \}_{0 \leq t < T})$ be an evolving surface with a smooth boundary. We say that $\Gamma^\varepsilon (t)$ is a variation of $\Gamma (t)$ if $\Gamma^\varepsilon (0) = \Gamma_0 $ and $\Gamma^\varepsilon |_{\varepsilon = 0} (t) = \Gamma (t) $. Set
\begin{equation*}
\Gamma^\varepsilon_T = \bigcup_{0 < t < T}\{ \Gamma^\varepsilon (t) \times \{ t \} \} \text{ and }\overline{\Gamma^\varepsilon_T} = \bigcup_{0 \leq t < T}\{ \overline{\Gamma^\varepsilon (t)} \times \{ t \} \}.
\end{equation*}
Let $\widetilde{x}^\varepsilon = \widetilde{x}^\varepsilon ( \xi , t )$ be a flow map on $\Gamma^\varepsilon (t)$, and $v^\varepsilon = v^\varepsilon (x, t) = { }^t ( v^\varepsilon_1 , v^\varepsilon_2 , v^\varepsilon_3 )$ be the velocity determined by the flow map $\widetilde{x}^\varepsilon$, i.e., for $\xi \in \overline{\Gamma_0}$ and $0 \leq t  < T$,
\begin{equation*}
\begin{cases}
v^\varepsilon = v^\varepsilon (x , t) = { }^t (v_1^\varepsilon (x , t) ,v_2^\varepsilon (x , t) ,v_3^\varepsilon (x , t) ),\\
\widetilde{x}^\varepsilon = \widetilde{x}^\varepsilon ( \xi , t ) = ( \widetilde{x}_1^\varepsilon ( \xi , t )  , \widetilde{x}_2^\varepsilon ( \xi , t )  , \widetilde{x}_3^\varepsilon ( \xi , t )  ),\\
\widetilde{x}_t^\varepsilon ( \xi , t ) = v^\varepsilon ( \widetilde{x}^\varepsilon ( \xi , t ) , t),\\
\widetilde{x}^\varepsilon ( \xi , 0) = \xi .
\end{cases}
\end{equation*}
We say that $(  \widetilde{x}^\varepsilon , \Gamma^\varepsilon_T )$ is a variation of $( \widetilde{x} , \Gamma_T )$ if $\widetilde{x}^\varepsilon = \widetilde{x}^\varepsilon ( \xi, t )$ is smooth as a function of $(\varepsilon , \xi , t ) \in (-1 , 1) \times \overline{\Gamma_0} \times [0, T)$ and $\widetilde{x}^\varepsilon ( \xi , t) |_{\varepsilon = 0} = \widetilde{x} ( \xi , t)$. Assume that $\Gamma^\varepsilon (t)$ is expressed by
\begin{equation*}
\Gamma^\varepsilon (t) = \{ x = { }^t ( x_1 , x_2, x_3 ) \in \mathbb{R}^3; { \ } x = \widetilde{x}^\varepsilon ( \xi , t), { \ }\xi \in \Gamma_0 \} .
\end{equation*}
For every $0 \leq t < T$ and $X \in \overline{U} $, 
\begin{equation*}
\widehat{x}^\varepsilon (X , t) := \widetilde{x}^\varepsilon ( \Phi (X) , t ) (= \widetilde{x}^\varepsilon ( \xi , t ) ).
\end{equation*}
Note that $\Gamma_0 = \{ \xi \in \mathbb{R}^3;{ \ }\xi = \Phi (X), X \in U \}$ (see Definition \ref{def22}). Moreover, for every $0 \leq t < T$ and $X \in \overline{U} $,
\begin{multline*}
\mathfrak{g}^\varepsilon_1 = \mathfrak{g}^\varepsilon_1 (X,t) := \frac{\partial \widehat{x}^\varepsilon}{\partial X_1},{ \ }\mathfrak{g}^\varepsilon_2 = \mathfrak{g}^\varepsilon_2 (X,t) := \frac{\partial \widehat{x}^\varepsilon}{\partial X_2},{ \ }\mathfrak{g}^\varepsilon_{\alpha \beta} := \mathfrak{g}^\varepsilon_\alpha \cdot \mathfrak{g}^\varepsilon_\beta,\\
( \mathfrak{g}_\varepsilon^{ \alpha \beta })_{2 \times 2} := ( \mathfrak{g}^\varepsilon_{\alpha \beta })_{2 \times 2}^{-1} = \frac{1}{\mathfrak{g}^\varepsilon_{11} \mathfrak{g}^\varepsilon_{22} - \mathfrak{g}^\varepsilon_{12} \mathfrak{g}^\varepsilon_{21} } 
\begin{pmatrix}
\mathfrak{g}^\varepsilon_{22} & - \mathfrak{g}^\varepsilon_{21}\\
- \mathfrak{g}^\varepsilon_{12} & \mathfrak{g}^\varepsilon_{11}
\end{pmatrix},{ \ }\mathfrak{g}_\varepsilon^\alpha := \mathfrak{g}_\varepsilon^{\alpha \beta} \mathfrak{g}^\varepsilon_\beta,\\
\acute{ \mathfrak{g}}^\varepsilon_{\alpha } :=  \frac{d}{d t} \mathfrak{g}^\varepsilon_\alpha, { \ }\acute{\mathfrak{g}}^\varepsilon_{\alpha \beta} := \frac{d }{d t} \mathfrak{g}^\varepsilon_{\alpha \beta }  , { \ }  \mathcal{G}^\varepsilon = \mathcal{G}^\varepsilon (X,t) := | \mathfrak{g}^\varepsilon_1 \times \mathfrak{g}^\varepsilon_2 |^2 = \mathfrak{g}^\varepsilon_{11} \mathfrak{g}^\varepsilon_{22} - \mathfrak{g}^\varepsilon_{12} \mathfrak{g}^\varepsilon_{21}. 
\end{multline*}
Assume that for all $-1 < \varepsilon < 1$, $X \in \overline{U}$, and $0 \leq t < T$,
\begin{equation*}
\mathfrak{g}^\varepsilon_1 \times \mathfrak{g}_2^\varepsilon  \geq \lambda_2 > 0,
\end{equation*}
where $\lambda_2$ is the positive constant appearing in Definition \ref{def23}. From now on we assume that $\rho^\varepsilon = \rho^\varepsilon ( x , t)$ and $v^\varepsilon = v^\varepsilon (x, t) $ are $C^2$-functions on $\overline{\Gamma_T^\varepsilon}$.

To consider our action integral, we state the surface transport theorem and its application.
\begin{proposition}[Surface transport theorem and its application]\label{prop28}{ \ }\\
$(\mathrm{i})$ Let $f \in C^1 ( \Gamma_T)$. Then for every $\Omega (t) \subset \Gamma (t)$, 
\begin{equation*}
\frac{d }{d t} \int_{\Omega (t)} f (x,t) { \ }d \mathcal{H}^2_x = \int_{\Omega (t)} \{ D_t f + ({\rm{div}}_\Gamma v) f \} (x,t) { \ }d \mathcal{H}^2_x.
\end{equation*}
$(\mathrm{ii})$ Assume that for each $0 < t <T$ and $\Omega (t) \subset \Gamma (t)$,
\begin{equation*}
\frac{d}{d t} \int_{\Omega (t)} \rho (x,t) { \ }d \mathcal{H}^2_x = 0 .
\end{equation*}
Then
\begin{equation*}
D_t \rho + ({\rm{div}}_\Gamma v ) \rho = 0 \text{ on }\Gamma_T. 
\end{equation*}
$(\mathrm{iii})$ Let $\varepsilon \in (-1,1)$. Assume that for each $0 < t <T$ and $\Omega (t) \subset \Gamma^\varepsilon (t)$,
\begin{equation*}
\frac{d}{d t} \int_{\Omega (t)} \rho^\varepsilon (x,t) { \ }d \mathcal{H}^2_x = 0 .
\end{equation*}
Then
\begin{equation*}
D^\varepsilon_t \rho^\varepsilon + ({\rm{div}}_{\Gamma^\varepsilon} v^\varepsilon ) \rho^\varepsilon = 0 \text{ on }\Gamma^\varepsilon_T. 
\end{equation*}
Here
\begin{equation*}
D_t^\varepsilon f := \partial_t f + ( v^\varepsilon \cdot \nabla ) f .
\end{equation*}
\end{proposition}
\noindent Note that 
\begin{align*}
D_t f + ({\rm{div}}_\Gamma v) f & = D_t^N f + {\rm{div}}_\Gamma (f v),\\
D_t^\Gamma f + ({\rm{div}}_\Gamma v) f & = \partial_t f + {\rm{div}}_\Gamma (f v),
\end{align*}
where $D_t^N f = \partial_t f + ( v \cdot n ) (n \cdot \nabla )f $ and $D_t^\Gamma f = \partial_t f + ( v \cdot \nabla_\Gamma )f $. Note also that $D_t f = D_t^\Gamma f $ if $v \cdot n  =0$. The proof of Proposition \ref{prop28} in the case when the surface $\Gamma (t)$ is a closed surface can be founded in Betounes \cite{Bet86}, Gurtin-Struthers-Williams \cite{GSW89}, Dziuk-Elliott \cite{DE07}, and Koba-Liu-Giga \cite{KLG17}. This paper gives a sketch of the proof of Proposition \ref{prop28} for the readers.

\begin{theorem}[Variation of flow map to action integral]\label{thm29}{ \ }\\
Suppose that $(\widetilde{x}^\varepsilon , \Gamma^\varepsilon_T)$ is a variation of $(\widetilde{x}, \Gamma_T)$.  Let $\rho_0 \in C ( \overline{\Gamma_0}) $. Assume that $\rho$ and $\rho^\varepsilon$ satisfy
\begin{equation}\label{Eq210}
\begin{cases}
D_t \rho + ({\rm{div}}_\Gamma v ) \rho = 0{ \ }\text{ on }\Gamma_T,\\
\rho |_{t = 0} = \rho_0,
\end{cases}
\end{equation}
\begin{equation}\label{Eq211}
\begin{cases}
D_t^\varepsilon \rho^\varepsilon + ({\rm{div}}_{\Gamma^\varepsilon} v^\varepsilon ) \rho^\varepsilon = 0{ \ }\text{ on }\Gamma^\varepsilon_T,\\
\rho^\varepsilon |_{t = 0} = \rho_0,
\end{cases}
\end{equation}
and that for every $\xi \in \overline{\Gamma_0}$ and $0 \leq t < T$,
\begin{equation*}
\rho^\varepsilon (\widetilde{x}^\varepsilon (\xi, t) , t) |_{\varepsilon = 0} = \rho ( \widetilde{x} ( \xi , t ) , t).
\end{equation*}
Suppose that there exists $z = z (x,t) = { }^t (z_1, z_2 , z_3 ) \in [ C^2 ( \overline{\Gamma_T}) ]^3$ such that for $0 \leq t <T$ and $\xi \in \overline{\Gamma_0}$
\begin{align*}
\frac{d \widetilde{x}^\varepsilon }{d \varepsilon} \bigg|_{\varepsilon = 0} ( \xi ,t ) = z ( \widetilde{x} ( \xi, t ), t ),\\
\lim_{\tau \uparrow T} z (\widetilde{x}(\xi , \tau ) , \tau ) = { }^t ( 0, 0 ,0).
\end{align*}
Then the following two assertions hold:\\
$(\mathrm{i})$ For each flow map $\widetilde{x}^\varepsilon = \widetilde{x}^\varepsilon ( \xi , t)$, set the \emph{action integral} as follows:
\begin{equation*}
Act[\widetilde{x}^\varepsilon ] = - \int_0^T \int_{\Gamma^\varepsilon (t)} \frac{1}{2} \rho^\varepsilon (x,t) |v^\varepsilon (x,t) |^2 { \ }d \mathcal{H}^2_x d t.
\end{equation*}
Then
\begin{equation*}
\frac{d}{d \varepsilon} \bigg|_{\varepsilon = 0} Act [ \widetilde{x}^\varepsilon ] = \int_0^T \int_{\Gamma (t)} \{\rho D_t v \} (x,t) \cdot z (x, t) { \ }d \mathcal{H}^2_x d t.
\end{equation*}
Moreover, if $z \cdot n = 0$ on $\Gamma_T$, then
\begin{equation*}
\frac{d}{d \varepsilon} \bigg|_{\varepsilon = 0} Act [ \widetilde{x}^\varepsilon ] = \int_0^T \int_{\Gamma (t)} \{ P_\Gamma \rho D_t v \} (x,t) \cdot z (x, t) { \ }d \mathcal{H}^2_x d t.
\end{equation*}
$(\mathrm{ii})$ Let $p$ be a $C^1$-function. For each flow map $\widetilde{x}^\varepsilon = \widetilde{x}^\varepsilon ( \xi , t)$, set the \emph{action integral} as follows:
\begin{equation*}
A_B[\widetilde{x}^\varepsilon ] = - \int_0^T \int_{\Gamma^\varepsilon (t)} \left\{ \frac{1}{2} \rho^\varepsilon (x,t) |v^\varepsilon (x,t) |^2 - p (\rho^\varepsilon (x,t) ) \right\} { \ }d \mathcal{H}^2_x d t.
\end{equation*}
If
\begin{equation*}
z \cdot \nu |_{\partial \Gamma_T} = 0,
\end{equation*}
then
\begin{equation*}
\frac{d}{d \varepsilon} \bigg|_{\varepsilon = 0} A_B [ \widetilde{x}^\varepsilon ] = \int_0^T \int_{\Gamma (t)} \{\rho D_t v + {\rm{grad}}_\Gamma \mathfrak{p} + \mathfrak{p} H_\Gamma n \} (x,t) \cdot z (x, t) { \ }d \mathcal{H}^2_x d t.
\end{equation*}
Moreover, if $z \cdot n = 0$ on $\Gamma_T$, then
\begin{equation*}
\frac{d}{d \varepsilon} \bigg|_{\varepsilon = 0} A_B [ \widetilde{x}^\varepsilon ] = \int_0^T \int_{\Gamma (t)} \{ P_\Gamma \rho D_t v + {\rm{grad}}_\Gamma \mathfrak{p} \} (x,t) \cdot z (x, t) { \ }d \mathcal{H}^2_x d t.
\end{equation*}
Here $\mathfrak{p}  = \mathfrak{p}(\rho ) := \rho p' (\rho)  - p (\rho) $.
\end{theorem}

Applying an energetic variational approach and the first law of thermodynamics with Theorems \ref{thm27} and \ref{thm29}, we can derive four systems \eqref{eq11}, \eqref{Eq113}, \eqref{Eq116}, and \eqref{Eq117}. Note that we derive the barotropic compressible fluid systems \eqref{Eq116} and \eqref{Eq117} by using the action integral $A_B$. See Section \ref{sect5} for detail.

Finally, we state the boundary conditions for our fluid systems.
\begin{theorem}[Conservation laws]\label{Thm210}{ \ }\\
\noindent $(\mathrm{i})$ Assume that \eqref{eq15} holds. Then any solution to system \eqref{eq11} satisfies \eqref{eq19} and \eqref{Eq110}.\\
\noindent $(\mathrm{ii})$ Assume that \eqref{eq15} and \eqref{eq16} hold. Then any solution to system \eqref{eq11} satisfies \eqref{Eq111}.\\
\noindent $(\mathrm{iii})$ Assume that \eqref{eq17} holds. Then any solution to system \eqref{eq11} satisfies \eqref{Eq112}.

\end{theorem}

\begin{theorem}[Energy laws]\label{Thm211}{ \ }\\
\noindent $(\mathrm{i})$ If either \eqref{eq15} or \eqref{Eq114} holds, then any solution to system \eqref{eq11} satisfies \eqref{Eq115}.\\
\noindent $(\mathrm{ii})$ If either \eqref{eq15} or \eqref{Eq114} holds, then any solution to system \eqref{Eq113} satisfies \eqref{Eq115}.\\
\noindent $(\mathrm{iii})$ If \eqref{Eq114} holds, then any solution to system \eqref{Eq116} satisfies \eqref{Eq118}.\\
\noindent $(\mathrm{iv})$ If \eqref{Eq114} holds, then any solution to system \eqref{Eq117} satisfies \eqref{Eq118}.
\end{theorem}

\begin{remark}\label{Rem212}
Instead of \eqref{eq15} or \eqref{Eq114}, we may consider
\begin{equation*}
\int_{\partial \Gamma (t)} S_\Gamma (v,\sigma ) \nu { \ }d \mathcal{H}^1_x = { }^t (0,0,0) \text{ or } \nu \cdot z|_{\partial \Gamma (t)} = 0.
\end{equation*}
\end{remark}

In Section \ref{sect3} we use a flow map and the Riemannian metric induced by the flow map to prove Theorem \ref{thm24} and Proposition \ref{prop28}. In Section \ref{sect4} we calculate variations of our dissipation energies and action integrals to prove Theorems \ref{thm27} and \ref{thm29}. In Section \ref{sect5} we apply an energetic variational approach to make mathematical models for compressible fluid flow on the evolving surface $\Gamma (t)$ with a boundary, and investigate the boundary conditions for our compressible fluid systems to prove Theorems \ref{Thm210} and \ref{Thm211}.

\section{Representation of Energy Densities}\label{sect3}

In this section we first study the representation of our energy densities (Assumption \ref{ass12}) to show the invariance of the energy densities throughout the Riemannian metric induced by flow maps. Then we give a  sketch of the proof of the surface transport theorem. To investigate our energy densities, we first prepare the following lemma.
\begin{lemma}\label{lem31}
$(\mathrm{i})$ For every $\psi \in C ( \mathbb{R}^3)$,
\begin{equation}\label{eq31}
\int_{\Gamma (t)} \psi (x) { \ }d \mathcal{H}^2_x = \int_{U} \widehat{\psi} (X,t) \sqrt{ \mathcal{G} (X,t )} { \ } d X .
\end{equation}
$(\mathrm{ii})$ For each $i, j = 1,2,3$,
\begin{equation}\label{eq32}
\int_{\Gamma (t)} [P_\Gamma]_{i j} { \ } d \mathcal{H}^2_x = \int_U \frac{\partial \widehat{x}_i }{\partial X_\alpha} \frac{\partial \widehat{x}_j }{\partial X_\beta} \mathfrak{g}^{\alpha \beta} \sqrt{ \mathcal{G} } { \ } d X.
\end{equation}
$(\mathrm{iii})$ For each $j=1,2,3$, and $\psi \in C^1 ( \mathbb{R}^3)$,
\begin{equation}
\int_{\Gamma (t)} \partial^\Gamma_j \psi { \ } d \mathcal{H}^2_x = \int_U \frac{\partial \widehat{x}_j }{\partial X_\alpha} \frac{\partial \widehat{\psi} }{\partial X_\beta}  \mathfrak{g}^{\alpha \beta } \sqrt{ \mathcal{G} }{ \ }d X.\label{eq33}
\end{equation}
$(\mathrm{iv})$ For each $\varphi = { }^t (\varphi_1 , \varphi_2 , \varphi_3 ) \in [C^1( \mathbb{R}^3)]^3$,
\begin{equation}
\int_{\Gamma (t)} {\rm{div}}_\Gamma \varphi { \ } d \mathcal{H}^2_x = \int_U \mathfrak{g}^\alpha \cdot \frac{\partial \widehat{\varphi}}{\partial X_\alpha} \sqrt{ \mathcal{G} }{ \ }d X.\label{eq34}
\end{equation}
$(\mathrm{v})$ For each $i, j=1,2,3$, and $\psi \in C^2 ( \mathbb{R}^3)$,
\begin{equation}
\int_{\Gamma (t)} \partial_i^\Gamma \partial^\Gamma_j \psi { \ } d \mathcal{H}^2_x = \int_U \mathfrak{g}^{\zeta \eta} \frac{\partial \widehat{x}_i}{\partial X_\zeta} \frac{\partial }{\partial X_\eta} \left( \frac{\partial \widehat{x}_j }{\partial X_\alpha} \frac{\partial \widehat{\psi} }{\partial X_\beta}  \mathfrak{g}^{\alpha \beta } \right) \sqrt{ \mathcal{G} }{ \ }d X.\label{eq35}
\end{equation}
$(\mathrm{vi})$ For each $\psi \in C^2( \mathbb{R}^3)$ and $\kappa \in C^1 (\mathbb{R}^3)$,
\begin{equation}
\int_{\Gamma (t)} {\rm{div}}_\Gamma (\kappa \nabla_\Gamma \psi ) { \ } d \mathcal{H}^2_x = \int_U \left\{ \frac{1}{\sqrt{\mathcal{G}}} \frac{\partial}{\partial X_\alpha} \left( \widehat{\kappa} \sqrt{\mathcal{G}} \mathfrak{g}^{\alpha \beta} \frac{\partial \widehat{\psi}}{\partial X_\beta} \right) \right\}  \sqrt{ \mathcal{G} }{ \ }d X.\label{eq36}
\end{equation}
Here $\widehat{\psi} = \widehat{\psi} (X,t) : = \psi ( \widehat{x} (X,t))$, $\widehat{\varphi} = \widehat{\varphi} (X,t) : = \varphi ( \widehat{x} (X,t))$, and\\ $\widehat{\kappa} = \widehat{\kappa} (X,t) : = \kappa ( \widehat{x} (X,t))$.
\end{lemma}
\begin{proof}[Proof of Lemma \ref{lem31}]
We first show $(\mathrm{i})$. Fix $\psi \in C ( \mathbb{R}^3)$. Since
\begin{equation*}
\Gamma (t) = \{ x \in \mathbb{R}^3; { \ } x = \widehat{x} ( \Phi (X), t), { \ }X \in U \},
\end{equation*}
it follows from the formula of the surface integral to see that
\begin{equation*}
\int_{\Gamma (t)} \psi (x) { \ }d \mathcal{H}^2_x = \int_{U} \psi (\widehat{x} (X, t) , t) { \ } \left| \frac{\partial \widehat{x}}{\partial X_1} \times \frac{\partial \widehat{x}}{\partial X_2} \right|  { \ }d X = \int_{U} \widehat{\psi} { \ }\sqrt{ \mathcal{G}} { \ }d X.
\end{equation*}
Therefore we have \eqref{eq31}. Next we prove $(\mathrm{ii})$. A direct calculation shows that
\begin{multline}\label{eq37}
\int_{\Gamma (t)} n (x,t) { \ } d \mathcal{H}^2_x = \int_{\Gamma (t)} 
\begin{pmatrix}
n_1\\
n_2\\
n_3
\end{pmatrix}
{ \ }d \mathcal{H}^2_x  = \int_{U} \frac{ \mathfrak{g}_1 \times \mathfrak{g}_2}{ | \mathfrak{g}_1 \times \mathfrak{g}_2 | } \sqrt{ \mathcal{G}}{ \ } d X\\
 = \int_{U} \frac{1}{ \sqrt{\mathfrak{g}_{11} \mathfrak{g}_{22} - \mathfrak{g}_{12} \mathfrak{g}_{21}     }  }
\begin{pmatrix}
\frac{\partial \widehat{x}_2 }{\partial X_1} \frac{\partial \widehat{x}_3 }{\partial X_2} - \frac{\partial \widehat{x}_2 }{\partial X_2} \frac{\partial \widehat{x}_3 }{\partial X_1}\\
\frac{\partial \widehat{x}_3 }{\partial X_1} \frac{\partial \widehat{x}_1 }{\partial X_2} - \frac{\partial \widehat{x}_3 }{\partial X_2} \frac{\partial \widehat{x}_1 }{\partial X_1}\\
\frac{\partial \widehat{x}_1 }{\partial X_1} \frac{\partial \widehat{x}_2 }{\partial X_2} - \frac{\partial \widehat{x}_1 }{\partial X_2} \frac{\partial \widehat{x}_2 }{\partial X_1}
\end{pmatrix}
 \sqrt{ \mathcal{G}}{ \ } d X.
\end{multline}
Using \eqref{eq37}, \eqref{eq21}, and
\begin{equation}\label{eq38}
\begin{pmatrix}
\mathfrak{g}^{11} & \mathfrak{g}^{12}\\
\mathfrak{g}^{21} & \mathfrak{g}^{22}
\end{pmatrix}
= \frac{1}{ \mathfrak{g}_{11} \mathfrak{g}_{22} - \mathfrak{g}_{12} \mathfrak{g}_{21} }
\begin{pmatrix}
\mathfrak{g}_{22} & - \mathfrak{g}_{12}\\
- \mathfrak{g}_{21} & \mathfrak{g}_{11}
\end{pmatrix},
\end{equation}
we check that for each $i,j=1,2,3$,
\begin{equation*}
\int_{\Gamma (t)} [P_\Gamma]_{i j} { \ } d \mathcal{H}^2_x = \int_{\Gamma (t)} (\delta_{i j} - n_i n_j ) { \ } d \mathcal{H}^2_x = \int_U \frac{\partial \widehat{x}_i }{\partial X_\alpha} \frac{\partial \widehat{x}_j }{\partial X_\beta} \mathfrak{g}^{\alpha \beta} \sqrt{ \mathcal{G} } { \ } d X,
\end{equation*}
which is \eqref{eq32}. Now, we derive $(\mathrm{iii})$ and $(\mathrm{iv})$. Fix $j=1,2,3$, $\psi \in C^1 ( \mathbb{R}^3)$, $\varphi \in [C^1 (\mathbb{R}^3)]^3$. Since 
\begin{equation*}
\frac{\partial \widehat{\psi}}{\partial X_\beta} = \frac{\partial}{\partial X_\beta} \psi ( \widehat{x} ( X , t) , t ) = \frac{\partial \widehat{x}_i}{\partial X_\beta} \frac{\partial \psi}{\partial \widehat{x}_i},
\end{equation*}
we use $\nabla_\Gamma \psi = P_\Gamma \nabla \psi $ and \eqref{eq32} to observe that
\begin{align*}
\int_{\Gamma (t)} \partial^\Gamma_j \psi { \ } d \mathcal{H}^2_x & = \int_{\Gamma (t)} [P_\Gamma]_{j i}\partial_i  \psi { \ }d \mathcal{H}^2_x\\
& = \int_U \frac{\partial \widehat{x}_j }{\partial X_\alpha} \frac{\partial \widehat{x}_i }{\partial X_\beta}  \mathfrak{g}^{\alpha \beta } \frac{\partial \psi}{\partial \widehat{x}_i} \sqrt{ \mathcal{G} }{ \ }d X\\
& = \int_U \frac{\partial \widehat{x}_j }{\partial X_\alpha} \frac{\partial \widehat{\psi} }{\partial X_\beta} \mathfrak{g}^{\alpha \beta } \sqrt{ \mathcal{G} }{ \ }d X.
\end{align*}
This is \eqref{eq33}. From \eqref{eq33}, we check that
\begin{multline*}
\int_{\Gamma (t)} {\rm{div}}_\Gamma \varphi { \ } d \mathcal{H}^2_x = \int_{\Gamma (t)} \partial_j^\Gamma \varphi_j { \ } d \mathcal{H}^2_x = \int_U \frac{\partial \widehat{x}_j }{\partial X_\alpha} \frac{\partial \widehat{\varphi}_j }{\partial X_\beta} \mathfrak{g}^{\alpha \beta } \sqrt{ \mathcal{G} }{ \ }d X\\
 = \int_U \left( \mathfrak{g}_\alpha \cdot \frac{\partial \widehat{ \varphi} }{\partial X_\beta} \right) \mathfrak{g}^{\alpha \beta} \sqrt{ \mathcal{G} }{ \ }d X = \int_U \left( \mathfrak{g}^\beta \cdot \frac{\partial \widehat{\varphi}}{\partial X_\beta} \right) \sqrt{ \mathcal{G} }{ \ }d X.
\end{multline*}
Note $\mathfrak{g}^{\beta} = \mathfrak{g}^{\alpha \beta} \mathfrak{g}_\beta$. Applying \eqref{eq33}, we see $(\mathrm{v})$ and $(\mathrm{vi})$. Therefore the lemma follows.
  \end{proof}
Let us attack Theorem \ref{thm24}.
\begin{proof}[Proof of Theorem \ref{thm24}]
We first show \eqref{eq23}. A direct calculation shows that
\begin{align}
\acute{\mathfrak{g}}_{\alpha \beta} = \frac{\partial v}{\partial X_\alpha} \cdot \frac{\partial \widehat{x} }{\partial X_\beta} + \frac{\partial \widehat{x} }{\partial X_\alpha} \cdot \frac{\partial v }{\partial X_\beta} & = \frac{\partial v_j}{\partial X_\alpha} \frac{\partial \widehat{x}_j }{\partial X_\beta} + \frac{\partial \widehat{x}_j }{\partial X_\alpha} \frac{\partial v_j }{\partial X_\beta} \notag \\
& = \frac{\partial \widehat{x}_i}{\partial X_\alpha} \frac{\partial v_j}{\partial \widehat{x}_i} \frac{\partial \widehat{x}_j }{\partial X_\beta} + \frac{\partial \widehat{x}_j }{\partial X_\alpha} \frac{\partial \widehat{x}_i}{\partial X_\beta} \frac{\partial v_j }{\partial \widehat{x}_i}. \label{eq39}
\end{align}
From \eqref{eq39}, the Einstein summation convention, and $\mathfrak{g}^{21} = \mathfrak{g}^{12}$ , we find that
\begin{equation}\label{Eq310}
\acute{ \mathfrak{g}}_{\alpha \beta} \mathfrak{g}^{\alpha \beta} = 2 \frac{\partial \widehat{x}_j }{\partial X_\alpha} \frac{\partial \widehat{x}_i}{\partial X_\beta} \frac{\partial v_j }{\partial \widehat{x}_i} \mathfrak{g}^{ \alpha \beta } =  2 \frac{\partial \widehat{x}_j }{\partial X_\alpha} \frac{\partial v_j}{\partial X_\beta} \mathfrak{g}^{ \alpha \beta }. 
\end{equation}
Using \eqref{eq33} and \eqref{Eq310}, we check that
\begin{align}
\int_{\Gamma (t)} {\rm{div}}_\Gamma v { \ }d \mathcal{H}^2_x = \int_{\Gamma (t)} \partial^\Gamma_j v_j { \ } d \mathcal{H}^2_x & = \int_U \frac{\partial \widehat{x}_j }{\partial X_\alpha} \frac{\partial v_j }{\partial X_\beta}  \mathfrak{g}^{\alpha \beta } \sqrt{ \mathcal{G} }{ \ }d X \notag \\
&= \int_U \frac{1}{2} \acute{ \mathfrak{g} }_{\alpha \beta} \mathfrak{g}^{\alpha \beta} \sqrt{ \mathcal{G} } { \ }d X.\label{Eq311}
\end{align}
This gives \eqref{eq23}. 

Secondly, we derive \eqref{eq22}. From \eqref{eq38}, $\mathcal{G} = \mathfrak{g}_{11} \mathfrak{g}_{22} - \mathfrak{g}_{12} \mathfrak{g}_{21}$, and $\mathfrak{g}_{12} = \mathfrak{g}_{21}$, we observe that
\begin{align*}
\int_{U} \frac{d \sqrt{ \mathcal{G}} }{d t} { \ }d X & = \int_{\Gamma (t)} \frac{1}{2 \mathcal{G}} \left( \frac{d \mathcal{G}}{d t} \right)\sqrt{ \mathcal{G}} { \ }d X\\
& = \int_{U} \frac{1}{2} \acute{ \mathfrak{g}}_{\alpha \beta} \mathfrak{g}^{ \alpha \beta} \sqrt{ \mathcal{G}} { \ }d X.
\end{align*}
From \eqref{Eq311} we have \eqref{eq22}.

Thirdly, we deduce \eqref{eq25}. Using \eqref{Eq311}, we see that
\begin{align*}
\int_{\Gamma (t)} \frac{1}{2} e_2 \left( \frac{1}{4} \acute{ \mathfrak{g}}_{\alpha \beta} \acute{ \mathfrak{g}}_{\zeta \eta} \mathfrak{g}^{\alpha \beta} \mathfrak{g}^{\zeta \eta } \right) \sqrt{ \mathcal{G}} { \ }d X & = \int_{\Gamma (t)} \frac{1}{2} e_2 \left( \frac{1}{2} \acute{ \mathfrak{g}}_{\alpha \beta}  \mathfrak{g}^{\alpha \beta} \frac{1}{2} \acute{ \mathfrak{g}}_{\zeta \eta}\mathfrak{g}^{\zeta \eta } \right) \sqrt{ \mathcal{G}} { \ }d X\\
& = \int_{\Gamma (t)} \frac{1}{2} e_2 ( | {\rm{div}}_\Gamma v |^2) { \ }d \mathcal{H}^2_x,
\end{align*}
which is \eqref{eq25}.

Fourthly, we show \eqref{eq26} and \eqref{eq27}. A direct calculation shows that
\begin{multline}
\left( \frac{\partial \widehat{x}_j }{\partial X_\alpha} \frac{\partial \widehat{\theta} }{\partial X_\beta}  \mathfrak{g}^{\alpha \beta } \right) \left( \frac{\partial \widehat{x}_j }{\partial X_\zeta} \frac{\partial \widehat{\theta} }{\partial X_\eta}  \mathfrak{g}^{\zeta \eta } \right) = \mathfrak{g}_{\alpha \zeta} \mathfrak{g}^{\zeta \eta} \mathfrak{g}^{\alpha \beta} \frac{\partial \widehat{\theta} }{\partial X_\beta} \frac{\partial \widehat{\theta} }{\partial X_\eta} \\
= \delta_{\alpha \zeta} \mathfrak{g}^{\alpha \beta} \frac{\partial \widehat{\theta} }{\partial X_\beta} \frac{\partial \widehat{\theta} }{\partial X_\eta} = \mathfrak{g}^{\alpha \beta} \frac{\partial \widehat{\theta} }{\partial X_\alpha} \frac{\partial \widehat{\theta} }{\partial X_\beta}. \label{Eq312}
\end{multline}
Here $\delta_{\alpha \zeta}$ is the Kronecker delta. By \eqref{eq33} and \eqref{Eq312}, we observe that
\begin{align*}
\int_{\Gamma (t)} \frac{1}{2} e_3 ( | {\rm{grad}}_\Gamma \theta |^2) { \ }d \mathcal{H}^2_x & = \int_{\Gamma (t)} \frac{1}{2} e_3 ( | ( \partial_j^\Gamma \theta ) (\partial_j^\Gamma \theta) |^2) { \ }d \mathcal{H}^2_x\\
& = \int_U \frac{1}{2} e_3 \left( \mathfrak{g}^{\alpha \beta} \frac{\partial \widehat{\theta}}{\partial X_\alpha} \frac{\partial \widehat{\theta}}{\partial X_\beta}  \right) \sqrt{ \mathcal{G}}{ \ }d X.
\end{align*}
Thus, we have \eqref{eq26}. Similarly, we obtain \eqref{eq27}.

Finally, we derive \eqref{eq24}. For each $i,j = 1 , 2, 3$, write
\begin{equation*}
[\widehat{D} (v)]_{i j} = \frac{1}{2} \left( \frac{\partial v_j}{\partial \widehat{x}_i} + \frac{\partial v_i}{\partial \widehat{x}_j} \right).
\end{equation*}
From \eqref{eq39} we find that
\begin{equation}\label{Eq313}
\frac{1}{2} \acute{\mathfrak{g}}_{\alpha \beta} = \frac{\partial \widehat{x}_i}{\partial X_\alpha} [\widehat{D} (v) ]_{i j} \frac{\partial \widehat{x}_j}{\partial X_\beta}.
\end{equation}
By \eqref{eq32} and \eqref{Eq313}, we observe that
\begin{align*}
\frac{1}{4} \acute{ \mathfrak{g}}_{\alpha \beta} \acute{ \mathfrak{g}}_{\zeta \eta} \mathfrak{g}^{\alpha \zeta} \mathfrak{g}^{\beta \eta} & = \frac{\partial \widehat{x}_i}{\partial X_\alpha} [\widehat{D} (v) ]_{i j} \frac{\partial \widehat{x}_j}{\partial X_\beta}\mathfrak{g}^{\alpha \zeta} \mathfrak{g}^{\beta \eta} \frac{\partial \widehat{x}_k}{\partial X_\zeta} [\widehat{D} (v) ]_{k \ell} \frac{\partial \widehat{x}_\ell}{\partial X_\eta}\\
& = [\widehat{D} (v) ]_{i j}  \frac{\partial \widehat{x}_i}{\partial X_\alpha} \frac{\partial \widehat{x}_k}{\partial X_\zeta} \mathfrak{g}^{\alpha \zeta} [\widehat{D} (v) ]_{k \ell} \frac{\partial \widehat{x}_j}{\partial X_\beta} \frac{\partial \widehat{x}_\ell}{\partial X_\eta} \mathfrak{g}^{\beta \eta}\\
& = [\widehat{D}(v)]_{i j} [P_\Gamma ]_{i k} [\widehat{D} (v)]_{k \ell} [P_\Gamma ]_{j \ell} .
\end{align*}
This implies that
\begin{align*}
\int_{U} \frac{1}{2} e_1 \left( \frac{1}{4} \acute{ \mathfrak{g} }_{\alpha \beta} \acute{ \mathfrak{g} }_{\zeta \eta} \mathfrak{g}^{\alpha \zeta} \mathfrak{g}^{\beta \eta } \right) \sqrt{ \mathcal{G} } { \ }d X & = \int_{\Gamma (t)} \frac{1}{2} e_1 ( | P_\Gamma D (v) P_\Gamma |^2) { \ }d \mathcal{H}^2_x\\
& = \int_{\Gamma (t)} \frac{1}{2} e_1 ( |D_\Gamma (v) |^2) { \ }d \mathcal{H}^2_x.
\end{align*}
Thus, we have \eqref{eq24}. Therefore Theorem \ref{thm24} is proved.   \end{proof}

Let us give a sketch of the proof of Proposition \ref{prop28}.
\begin{proof}[Proof of Proposition \ref{prop28}]
We first show $(\mathrm{i})$. Let $f \in C^1 ( \Gamma_T)$. Fix $\Omega (t) \subset \Gamma (t)$. By the property $(\mathrm{iii})$ in Definition \ref{def23}, there is $U_1 \subset U$ such that
\begin{equation*}
\Omega (t) = \{ x \in \mathbb{R}^3; { \ }x = \widehat{x} (X,t), { \ }X \in U_1 \}.
\end{equation*}
Using \eqref{eq22}, we see that
\begin{align*}
\frac{d}{d t} \int_{\Omega (t)} f (x,t) { \ }d\mathcal{H}^2_x & = \frac{d}{d t} \int_{U_1} f ( \widehat{x} (X,t) , t ) \sqrt{ \mathcal{G}(X,t) } { \ }d X\\
&= \int_{\Omega (t)} \{ D_t f + ({\rm{div}}_\Gamma v ) f  \}{ \ }d\mathcal{H}^2_x.
\end{align*}
Therefore the assertion $(\mathrm{i})$ is proved. Using the assertion $(\mathrm{i})$, we can derive $(\mathrm{ii})$ and $(\mathrm{iii})$. Therefore Proposition \ref{prop28} is proved.
  \end{proof}

\section{Variation of Dissipation Energies and Action Integral}\label{sect4}

We first calculate a variation of our dissipation energies with respect to the velocity to prove Theorem \ref{thm27}. Then we consider a variation of our action integral with respect to the flow map to prove Theorem \ref{thm29}.
\begin{proof}[Proof of Theorem \ref{thm27}]
We now prove $(\mathrm{i})$ and $(\mathrm{ii})$. Fix $\varphi = { }^t ( \varphi_1 , \varphi_2 , \varphi_3 ) \in [C_0^1 ( \Gamma (t) )]^3$. Applying the integration by parts \eqref{eq29}, we check that
\begin{multline*}
\frac{d}{d \varepsilon } \bigg|_{\varepsilon = 0} E_{D} [v + \varepsilon \varphi ] (t) \\
= - \int_{\Gamma (t)} \{ e_1'( |D_\Gamma (v) |^2 ) D_\Gamma (v): D_\Gamma ( \varphi) + e_2'( |{\rm{div}}_\Gamma v |^2 ) ({\rm{div}}_\Gamma v)({\rm{div}}_\Gamma \varphi ) \}{ \ }d \mathcal{H}^2_x,\\
= \int_{\Gamma (t)} {\rm{div}}_\Gamma \{ e_1'( |D_\Gamma (v) |^2 ) D_\Gamma (v) + e_2'( |{\rm{div}}_\Gamma v |^2 )  ({\rm{div}}_\Gamma v) P_\Gamma  \} \cdot \varphi { \ }d \mathcal{H}^2_x
\end{multline*}
and that
\begin{multline*}
\frac{d}{d \varepsilon } \bigg|_{\varepsilon = 0} E_{W} [v + \varepsilon \varphi ] (t) = \int_{\Gamma (t)} \{ ({\rm{div}}_\Gamma \varphi ) \sigma + \rho F \cdot \varphi \} { \ }d \mathcal{H}^2_x\\
= \int_{\Gamma (t)} \{ - {\rm{grad}}_\Gamma \sigma - \sigma H_\Gamma n + \rho F \} \cdot \varphi { \ }d \mathcal{H}^2_x = \int_{\Gamma (t)} \{ {\rm{div}}_\Gamma ( - \sigma P_\Gamma ) + \rho F \} \cdot \varphi { \ }d \mathcal{H}^2_x.
\end{multline*}
Here we used the facts that $\varphi |_{\partial \Gamma (t)} = { }^t (0,0,0)$, and $D_\Gamma (v) n = { }^t (0,0,0)$, and ${\rm{div}}_\Gamma \{ ({\rm{div}}_\Gamma v ) P_\Gamma  \} = {\rm{grad}}_\Gamma ({\rm{div}}_\Gamma v) + ({\rm{div}}_\Gamma v) H_\Gamma n$. Since $ \phi \cdot \varphi = P_\Gamma \phi \cdot \varphi$ for $\phi = { }^t (\phi_1 , \phi_2 , \phi_3 )$ if $\varphi \cdot n =0$ on $\Gamma (t)$, we see $(\mathrm{i})$ and $(\mathrm{ii})$. Similarly, we deduce $(\mathrm{iii})$ and $(\mathrm{iv})$. Therefore Theorem \ref{thm27} is proved.
 \end{proof}

To prove Theorem \ref{thm29} we study variations of the flow map $\widetilde{x}$. Let $(\widetilde{x}^\varepsilon , \Gamma^\varepsilon_T)$ be a variation of $( \widetilde{x} , \Gamma_T)$, that is, $\widetilde{x}^\varepsilon ( \xi ,t ) |_{\varepsilon = 0} = \widetilde{x} ( \xi ,t )$ and $ \widetilde{x}^\varepsilon ( \xi, 0) = \widetilde{x} (\xi , 0) = \xi$. For every $- 1< \varepsilon < 1 $, $\xi \in \overline{\Gamma_0}$, and $X \in \overline{U}$, set
\begin{align*}\widetilde{y}^\varepsilon ( \xi ,t ) & = \frac{d \widetilde{x}^\varepsilon}{d \varepsilon} ( \xi ,t ),\\
\widehat{y}^\varepsilon ( X ,t) & = \widetilde{y}^\varepsilon ( \Phi (X), t ) = \widetilde{y}^\varepsilon ( \xi , t ),\\
\widetilde{y} ( \xi ,t ) & = \frac{d \widetilde{x}^\varepsilon}{d \varepsilon} \bigg|_{\varepsilon = 0} ( \xi ,t ),\\
\widehat{y} ( X ,t) & = \widetilde{y} ( \Phi (X), t ) = \widetilde{y} ( \xi , t ).
\end{align*}
Assume that there is $z = z (x,t) = { }^t (z_1, z_2,z_3) \in [ C^2 ( \overline{\Gamma_T} ) ]^3$ such that for every $0 \leq t <T$ and $\xi \in \overline{\Gamma_0}$,
\begin{equation*}
z (x,t ) = z (\widetilde{x} (\xi, t) , t) = \widetilde{y} ( \xi , t ).
\end{equation*}
We now study fundamental properties of the variations $\widetilde{y}$, $\widehat{y}$, and $z$.
\begin{lemma}\label{lem41}
$(\mathrm{i})$ For every $X \in \overline{U}$ and $\xi \in \overline{\Gamma_0}$,
\begin{equation}\label{eq41}
\begin{cases}
\widehat{y} (X,0) = \widetilde{y} (\xi, 0)  = { }^t (0 , 0 , 0),\\
{\displaystyle{\lim_{\tau \uparrow T }\widehat{y} (X, \tau ) = \lim_{\tau \uparrow T} \widetilde{y} (\xi, \tau) }} = { }^t (0 , 0 , 0).
\end{cases}
\end{equation}
$(\mathrm{ii})$
\begin{equation}\label{eq42}
\int_{\Gamma (t)} {\rm{div}}_{\Gamma } z  { \ }d \mathcal{H}^2_x = \int_{U} \frac{d }{d \varepsilon} \bigg|_{\varepsilon = 0} \sqrt{\mathcal{G}^\varepsilon} { \ } d X .
\end{equation}
\end{lemma}

\begin{proof}[Proof of Lemma \ref{lem41}]
We first show $(\mathrm{i})$. Since $\widetilde{x}^\varepsilon ( \xi , 0)  = \xi $ for every $- 1 < \varepsilon < 1$ and $\xi \in \overline{\Gamma_0}$, we check that
\begin{align*}
\frac{d \widetilde{x}^\varepsilon}{d \varepsilon} ( \xi , 0 ) & = \lim_{h \to 0} \frac{\widetilde{x}^{\varepsilon + h} ( \xi , 0 ) - \widetilde{x}^\varepsilon ( \xi , 0) }{h}\\
& = \lim_{h \to 0} \frac{ \xi - \xi }{h} = { }^t (0 , 0 , 0 ).
\end{align*}
This implies that
\begin{align*}
\widehat{y} ( X , 0 ) = \widetilde{y} ( \xi , 0 ) = \frac{d \widetilde{x}^\varepsilon}{d \varepsilon} \bigg|_{\varepsilon = 0} ( \xi , 0 ) = { }^t (0 , 0 , 0 ).
\end{align*}
By the assumptions of Theorem \ref{thm29} on $z$ we see that
\begin{equation*}
{\displaystyle{\lim_{\tau \uparrow T }\widehat{y} (X, \tau ) = \lim_{\tau \uparrow T} \widetilde{y} (\xi, \tau) }} = { }^t (0 , 0 , 0).
\end{equation*}
Therefore we have \eqref{eq41}.

Next we show $(\mathrm{ii})$. We first show that 
\begin{equation}
 \frac{1}{\mathcal{G}^\varepsilon }\frac{d}{d \varepsilon} \mathcal{G}^\varepsilon = 2 \mathfrak{g}_\varepsilon^\alpha \cdot \frac{ \partial \widehat{y}^\varepsilon}{\partial X_\alpha}. \label{eq43}
\end{equation}
From $\mathcal{G}^\varepsilon = \mathfrak{g}_{11}^\varepsilon \mathfrak{g}_{22}^\varepsilon - \mathfrak{g}_{12}^\varepsilon \mathfrak{g}_{21}^\varepsilon$ and $\mathfrak{g}^\varepsilon_{\alpha \beta} = \mathfrak{g}^\varepsilon_\alpha \cdot \mathfrak{g}^\varepsilon_\beta$, we have
\begin{multline*}
\frac{d}{d \varepsilon} \mathcal{G}^\varepsilon = 2 \left( \mathfrak{g}^\varepsilon_1 \cdot \frac{\partial \widehat{y}^\varepsilon}{\partial X_1} \right) \mathfrak{g}^\varepsilon_{22} + 2 \left( \mathfrak{g}^\varepsilon_2 \cdot \frac{\partial \widehat{y}^\varepsilon}{\partial X_2} \right) \mathfrak{g}^\varepsilon_{11} \\- 2 \left( \mathfrak{g}^\varepsilon_1 \cdot \frac{\partial \widehat{y}^\varepsilon}{\partial X_2} \right) \mathfrak{g}^\varepsilon_{21} - 2 \left( \mathfrak{g}^\varepsilon_2 \cdot \frac{\partial \widehat{y}^\varepsilon}{\partial X_1} \right) \mathfrak{g}^\varepsilon_{12}.
\end{multline*}
Since $(\mathfrak{g}_\varepsilon^{\alpha \beta})_{2 \times 2} = (\mathfrak{g}^\varepsilon_{\alpha \beta})_{2 \times 2}^{-1}$ and $\mathfrak{g}_\varepsilon^{\beta \alpha} = \mathfrak{g}_\varepsilon^{\alpha \beta}$, we observe that
\begin{align*}
\frac{1}{ \mathcal{G}^\varepsilon}\frac{d}{d \varepsilon} \mathcal{G}^\varepsilon  & = 2 \left( (\mathfrak{g}^\varepsilon_1 \mathfrak{g}_\varepsilon^{11} + \mathfrak{g}^\varepsilon_2 \mathfrak{g}_\varepsilon^{12}) \cdot \frac{\partial \widehat{y}^\varepsilon}{\partial X_1} \right) + 2 \left( (\mathfrak{g}^\varepsilon_1 \mathfrak{g}^{21} + \mathfrak{g}^\varepsilon_2 \mathfrak{g}_\varepsilon^{22}) \cdot \frac{\partial \widehat{y}^\varepsilon}{\partial X_2} \right)\\
& =  2 \left( \mathfrak{g}_\varepsilon^1 \cdot \frac{\partial \widehat{y}^\varepsilon}{\partial X_1} \right) + 2 \left( \mathfrak{g}_\varepsilon^2 \cdot \frac{\partial \widehat{y}^\varepsilon}{\partial X_2} \right) = 2 \mathfrak{g}_\varepsilon^\alpha \cdot \frac{ \partial \widehat{y}^\varepsilon}{\partial X_\alpha}. 
\end{align*}
Thus, we have \eqref{eq43}. Using \eqref{eq43} and \eqref{eq34}, we check that
\begin{align*}
\int_{U} \frac{d }{d \varepsilon} \bigg|_{\varepsilon = 0} \sqrt{\mathcal{G}^\varepsilon} { \ } d X & = \int_{U} \left( \frac{1}{ 2 \mathcal{G}^\varepsilon} \frac{d \mathcal{G}^\varepsilon }{d \varepsilon} \right) \sqrt{\mathcal{G}^\varepsilon} \bigg|_{\varepsilon =0} { \ } d X \\
& = \int_{U} \left( \mathfrak{g}^\alpha \cdot \frac{ \partial \widehat{y}}{\partial X_\alpha} \right) \sqrt{\mathcal{G}} { \ } d X\\
& = \int_{\Gamma (t)} {\rm{div}}_\Gamma z { \ }d \mathcal{H}^2_x. 
\end{align*}
Note that $\widehat{y}(X,t) = z ( \widehat{x} (X,t) , t)$. Therefore the lemma follows. 
  \end{proof}

Let us attack Theorem \ref{thm29}.
\begin{proof}[Proof of Theorem \ref{thm29}]
We first show that $\rho$ and $\rho^\varepsilon$ are represented by
\begin{align}
\rho ( \widehat{x} (X,t) , t) & = \rho_0 ( \Phi (X)) \frac{ \sqrt{ \mathcal{G} (X,0)} }{ \sqrt{\mathcal{G} (X,t)} },\label{eq44}\\
\rho^\varepsilon ( \widehat{x}^\varepsilon (X,t) , t) & = \rho_0 ( \Phi (X)) \frac{ \sqrt{ \mathcal{G} (X,0)} }{ \sqrt{\mathcal{G}^\varepsilon (X,t)} },\label{eq45}
\end{align}
where $\Phi$ is the function in Definition \ref{def22}. To this end, we set
\begin{align*}
Q (X , t ) & = \rho ( \widehat{x} ( X , t ) , t ) \sqrt{ \mathcal{G} (X , t)},\\
Q^\varepsilon (X , t ) & = \rho^\varepsilon ( \widehat{x}^\varepsilon ( X , t ) , t ) \sqrt{ \mathcal{G}^\varepsilon (X , t)}.
\end{align*}
Fix $X$. A direct calculation gives
\begin{equation*}
\frac{d Q^\varepsilon }{d t} (X , t) = \left( \rho^\varepsilon_t + v_j^\varepsilon \frac{\partial \rho^\varepsilon}{\partial \widehat{x}^\varepsilon_j} \right) \sqrt{ \mathcal{G}^\varepsilon} + \rho^\varepsilon \frac{d \sqrt{\mathcal{G}^\varepsilon}}{dt}.
\end{equation*}
Using \eqref{eq22} and \eqref{Eq211}, we see that for each $U_1 \subset U$ and $0 < t <T$,
\begin{equation*}
\int_{U_1} \frac{d Q^\varepsilon}{d t} { \ }d X = \int_{\Gamma_1^\varepsilon (t)} \{ D_t^\varepsilon \rho^\varepsilon + ({\rm{div}}_{\Gamma^\varepsilon} v^\varepsilon ) \rho^\varepsilon \} { \ }d \mathcal{H}^2_x = 0,
\end{equation*}
where
\begin{equation*}
\Gamma_1^\varepsilon (t) = \{ x \in \mathbb{R}^3; { \ }x = \widehat{x}^\varepsilon (X, t), { \ }X \in U_1 \}.
\end{equation*}
This implies that
\begin{equation*}
\frac{d Q^\varepsilon}{d t} (X,t)= 0.
\end{equation*}
Integrating with respect to time, we see that for $0 < t <T$,
\begin{equation*}
Q^\varepsilon (X  , t) = Q^\varepsilon ( X , 0).
\end{equation*}
That is,
\begin{align*}
\rho^\varepsilon ( \widehat{x}^\varepsilon (X,t) , t) \sqrt{ \mathcal{G} (X,t)} & = \rho^\varepsilon ( \widehat{x}^\varepsilon (X,0) , 0) \sqrt{\mathcal{G} (X,0)}\\
 & = \rho_0 (\Phi (X)) \sqrt{ \mathcal{G} (X,0)}. 
\end{align*}
Thus, we see \eqref{eq45}. Similarly, we have \eqref{eq44}. Now we write
\begin{equation*}
\widetilde{\rho}_0 (X) = \rho_0 ( \Phi (X) ) \sqrt{ \mathcal{G} ( X , 0)}.
\end{equation*}
We now only prove $(\mathrm{ii})$ since the proof of $(\mathrm{i})$ is similar. By the definition of the action integral $A_B$ and \eqref{eq45}, we find that
\begin{multline}\label{eq46}
A_B[ \widetilde{x}^\varepsilon] = - \int_0^T \int_U \left\{ \frac{1}{2} \frac{ \widetilde{\rho}_0 (X) }{ \sqrt{ \mathcal{G}^\varepsilon } }  \frac{d \widehat{x}^\varepsilon }{d t} \cdot \frac{d \widehat{x}^\varepsilon}{d t} - p\left( \frac{ \widetilde{\rho}_0 (X) }{ \sqrt{ \mathcal{G}^\varepsilon }  } \right) \right\} \sqrt{ \mathcal{G}^\varepsilon } { \ }d X d t\\
= \int_0^T \int_U - \frac{1}{2} \widetilde{\rho}_0 (X)  \frac{d \widehat{x}^\varepsilon }{d t} \cdot \frac{d \widehat{x}^\varepsilon}{d t} { \ }d X d t + \int_0^T \int_U p\left( \frac{ \widetilde{\rho}_0 (X) }{ \sqrt{ \mathcal{G}^\varepsilon }  } \right)  \sqrt{ \mathcal{G}^\varepsilon } { \ }d X d t\\
=: A_1^\varepsilon + A_2^\varepsilon .
\end{multline}
Using the integration by parts with \eqref{eq41} and \eqref{eq44}, we check that
\begin{align}
\frac{d}{ d \varepsilon} \bigg|_{\varepsilon = 0} A_1^\varepsilon & = \int_0^T \int_U - \widetilde{\rho}_0 (X)  \frac{d \widehat{x} }{d t} \cdot \frac{d \widehat{y}}{d t} { \ }d X d t \notag \\
& = \int_0^T \int_U - \widetilde{\rho}_0 (X)  v ( \widehat{x} ( X, t ) , t ) \cdot \frac{d \widehat{y}}{d t} { \ }d X d t \notag \\
& = \int_0^T \int_U \frac{\widetilde{\rho}_0 (X)}{\sqrt{ \mathcal{G} (X,t) } } \left( v_t + v_j \frac{\partial v }{\partial \widehat{x}_j} \right) \cdot \widehat{y} \sqrt{\mathcal{G} (X,t)} { \ }d X d t \notag \\
&  = \int_0^T \int_{\Gamma (t)} \{ \rho D_t v \} (x,t ) \cdot z (x,t) { \ } d \mathcal{H}^2_x d t.\label{eq47}
\end{align}
Since
\begin{equation*}
\frac{d}{d \varepsilon} A_2^\varepsilon = \int_0^T \int_U  \left\{ - \frac{\widetilde{\rho}_0(X)}{\sqrt{\mathcal{G}^\varepsilon}} p \left( \frac{\widetilde{\rho}_0 (X)}{\sqrt{\mathcal{G}^\varepsilon} } \right) +  p \left( \frac{\widetilde{\rho}_0 (X)}{\sqrt{\mathcal{G}^\varepsilon} } \right) \right\} \frac{d \sqrt{\mathcal{G}^\varepsilon  }}{d \varepsilon} { \ } d X d t,
\end{equation*}
we use \eqref{eq42}, \eqref{eq44}, and \eqref{eq29} to check that
\begin{multline}\label{eq48}
\frac{d}{d \varepsilon} \bigg|_{\varepsilon = 0} A_2^\varepsilon = \int_0^T \int_{\Gamma (t)} \{ - \rho p' ( \rho ) + p (\rho ) \} {\rm{div}}_\Gamma z { \ } d \mathcal{H}_x^2 d t\\
= \int_0^T \int_{\Gamma (t)} {\rm{div}}_\Gamma ( \mathfrak{p} P_\Gamma ) \cdot z { \ } d \mathcal{H}_x^2 d t - \int_0^T \int_{\partial \Gamma (t)} \mathfrak{p} (\nu \cdot z)  { \ } d \mathcal{H}_x^1 d t.
\end{multline}
Here $\mathfrak{p} = \mathfrak{p} (\rho) = \rho p' ( \rho ) - p ( \rho )$. Since $z \cdot \nu |_{\partial \Gamma (t)} = 0$, we combine \eqref{eq46}, \eqref{eq47}, \eqref{eq48} to conclude that
\begin{equation*}
\frac{d}{d \varepsilon} \bigg|_{\varepsilon = 0} A_B [ \widetilde{x}^\varepsilon ] = \int_0^T \int_{\Gamma (t)} \{\rho D_t v + {\rm{grad}}_\Gamma \mathfrak{p} + \mathfrak{p} H_\Gamma n \} (x,t) \cdot z (x, t) { \ }d \mathcal{H}^2_x d t.
\end{equation*}
Therefore Theorem \ref{thm29} is proved.
  \end{proof}

\section{Mathematical Modeling}\label{sect5}

In this section we make mathematical models for compressible fluid flow on the evolving surface $\Gamma (t)$ with a boundary. In subsection \ref{subsec51} we apply both an energetic variational approach and the first law of thermodynamics to derive the generalized compressible fluid system \eqref{eq11}. In subsection \ref{subsec52} we study the enthalpy, the entropy, and the Helmholtz free energy of system \eqref{eq11}, and investigate the conservation and energy laws of the system. In subsection \ref{subsec53} we derive the generalized tangential compressible fluid system \eqref{Eq113} on the evolving surface by applying an energetic variational approach. In subsection \ref{subsec54} we state how to derive the two barotropic compressible fluid systems \eqref{Eq116} and \eqref{Eq117}.

\subsection{Derivation of Generalized Compressible Fluid System}\label{subsec51}

Let us derive the generalized compressible fluid system \eqref{eq11} on the evolving surface $\Gamma (t)$ with a boundary. We set the generalized energy densities for compressible fluid on the surface $\Gamma (t)$ as in Assumption \ref{ass12}, and the action integral $Act$, the energy $E_D$ dissipation due to the viscosities, the work $E_W$ done by the pressure and exterior force, the energy $E_{TD}$ dissipation due to thermal diffusion, and the energy $E_{GD}$ due to general diffusion as in Section \ref{sect2}. Based on Proposition \ref{prop28}, we obtain
\begin{equation}\label{eq51}
D_t \rho + ({\rm{div}}_\Gamma v ) \rho = 0 \text{ on } \Gamma_T .
\end{equation}

We first derive the momentum equation of system \eqref{eq11}. Set
\begin{equation*}
S_\Gamma (v , \sigma ) = e_1' (| D_\Gamma (v) |^2 ) D_\Gamma (v) + e_2' ( | {\rm{div}}_\Gamma v |^2 ) ({\rm{div}}_\Gamma v ) P_\Gamma - \sigma P_\Gamma .
\end{equation*}
From Theorems \ref{thm27} and \ref{thm29}, we have the following forces:
\begin{align*}
\frac{\delta Act}{\delta \widetilde{x}} & = \rho D_t v,\\
\frac{\delta E_{D + W}}{\delta v} & =  {\rm{div}}_\Gamma S_\Gamma (v , \sigma ) + \rho F .
\end{align*}
Here $E_{D+ W} = E_D + E_W$. We assume the following energetic variational principle:
\begin{equation*}
\frac{\delta Act}{\delta \widetilde{x}} = \frac{\delta E_{D + W}}{\delta v}
\end{equation*}
to have
\begin{equation}\label{eq52}
\rho D_t v = {\rm{div}}_\Gamma S_\Gamma ( v , \sigma ) + \rho F .
\end{equation}
Note that we may use $A_B$ instead of $Act$.

Secondly, we consider the internal energy. Applying the surface transport theorem with \eqref{eq51}, we see that
\begin{equation*}
\frac{d}{d t} \int_{\Gamma (t)} \frac{1}{2} \rho |v|^2 { \ }d \mathcal{H}^2_x = \int_{\Gamma (t)} \rho D_t v \cdot v { \ }d \mathcal{H}_x^2.
\end{equation*}
Using system \eqref{eq52} and \eqref{eq29}, we check that for $t_1 < t_2$,
\begin{multline}\label{eq53}
 \int_{\Gamma (t_2)} \frac{1}{2} \rho | v |^2 { \ }d \mathcal{H}^2_x + \int_{t_1}^{t_2} \int_{\Gamma ( \tau )} \{ \tilde{e}_D - ( {\rm{div}}_\Gamma v ) \sigma \} { \ }d \mathcal{H}^2_x d \tau\\
 = \int_{\Gamma (t_1)} \frac{1}{2} \rho | v |^2 { \ }d \mathcal{H}^2_x + \int_{t_1}^{t_2} \int_{\Gamma ( \tau )} \rho F \cdot v d \mathcal{H}^2_x d \tau + \int_{t_1}^{t_2} \int_{\partial \Gamma ( \tau )} B d \mathcal{H}^1_x d \tau. 
\end{multline}
Here
\begin{align}
\tilde{e}_D & = e_1'(|D_\Gamma (v)^2| ) |D_\Gamma (v)|^2 + e_2'(|{\rm{div}}_\Gamma v |^2 ) |{\rm{div}}_\Gamma v |^2,\label{eq54}\\
B & = \{ S_\Gamma (v , \sigma ) \nu \} \cdot v. \label{eq55}
\end{align}
Now we assume that $B \equiv 0$ and $F \equiv 0$. Then we have
\begin{equation*}
\int_{\Gamma (t_2)} \frac{1}{2} \rho | v |^2 { \ }d \mathcal{H}^2_x + \int_{t_1}^{t_2} \int_{\Gamma ( \tau )} \{ \tilde{e}_D - ( {\rm{div}}_\Gamma v ) \sigma \} { \ }d \mathcal{H}^2_x d \tau = \int_{\Gamma (t_1)} \frac{1}{2} \rho | v |^2 { \ }d \mathcal{H}^2_x. 
\end{equation*}
This implies that $\tilde{e}_D - ({\rm{div}}_\Gamma v ) \sigma $ is the density for the dissipation energy and the work done by pressure of our fluid system. From Theorem \ref{thm27} we have the following forces:
\begin{align}
\frac{\delta E_{TD}}{\delta \theta} & = {\rm{div}}_\Gamma \{ e'_3 (| {\rm{grad}}_\Gamma \theta |^2) {\rm{grad}}_\Gamma \theta \} ,\label{eq56}\\
\frac{\delta E_{GD}}{\delta C} & = {\rm{div}}_\Gamma \{ e'_4 (| {\rm{grad}}_\Gamma C |^2) {\rm{grad}}_\Gamma C \}. \label{eq57}
\end{align}
Applying the first law of thermodynamic, we obtain
\begin{equation}
\rho D_t e = {\rm{div}}_\Gamma q_\theta + \tilde{e}_D - ({\rm{div}}_\Gamma v ) \sigma \text{ on }\Gamma_T. \label{eq58}
\end{equation}
More precisely, we assume that for every $\Omega (t) \subset \Gamma (t)$
\begin{equation*}
\frac{d }{d t} \int_{\Omega (t)} \rho e { \ }d \mathcal{H}^2_x = \int_{\Omega (t)} \{ {\rm{div}}_\Gamma q_\theta + \tilde{e}_D - ({\rm{div}}_\Gamma v ) \sigma \} { \ }d \mathcal{H}^2_x.
\end{equation*}
Then we use \eqref{eq51} and the surface transport theorem to have \eqref{eq58}.

Finally, we derive the generalized diffusion system. We assume that the change of rate of the concentration $C$ equals to the force derived from a variation of the energy dissipation due to general diffusion, that is, for every $\Omega (t) \subset \Gamma (t)$, assume that
\begin{equation*}
\frac{d}{d t} \int_{\Omega (t)} C { \ } d \mathcal{H}^2_x = \int_{\Omega (t)} \frac{\delta E_{GD}}{\delta C} { \ }d \mathcal{H}^2_x.
\end{equation*}
Then we use the surface transport theorem to have
\begin{equation}
D_t C + ({\rm{div}}_\Gamma v ) C =  {\rm{div}}_\Gamma \{ e'_4 (| {\rm{grad}}_\Gamma C |^2) {\rm{grad}}_\Gamma C \} \text{ on } \Gamma_T.  \label{eq59}
\end{equation}
Combining \eqref{eq51}, \eqref{eq52}, \eqref{eq58}, and \eqref{eq59}, we therefore have the generalized compressible fluid system \eqref{eq11}.

\subsection{On Generalized Compressible Fluid System}\label{subsec52}

Let us study the generalized compressible fluid system \eqref{eq11}. We admit system \eqref{eq11}. Set the total energy $E_S = \rho |v|^2/2 + \rho e$ and
\begin{equation*}
D_t^N f = \partial_t f + (v \cdot n) (n \cdot \nabla ) f .
\end{equation*}
It is easy to check that \eqref{eq11} satisfies \eqref{eq12}.

We first investigate the enthalpy, the entropy, and the Helmholtz free energy of the compressible fluid system \eqref{eq11}. Assume that $\rho$ and $\theta$ are positive functions. Set the enthalpy $h = h (x,t)$ as follows $h = e + \sigma / \rho$. Then we have
\begin{equation*}
\rho D_t h = {\rm{div}}_\Gamma \{ e'_3 (|{\rm{grad}}_\Gamma \theta |^2) {\rm{grad}}_\Gamma \theta \} + \tilde{e}_D + D_t \sigma \text{ on }\Gamma_T.
\end{equation*}
Assume that the entropy $s = s (x,t)$ satisfy the Gibbs condition:$D_t e = \theta D_t s - \sigma D_t (1/\rho)$. Then we obtain
\begin{equation*}
\theta \rho D_t s = {\rm{div}}_\Gamma \{ e'_3 ( |{\rm{grad}}_\Gamma \theta |^2) {\rm{grad}}_\Gamma \theta \} + \tilde{e}_D \text{ on }\Gamma_T.
\end{equation*}
Set the Helmholtz free energy $e_F = e - \theta s$. A direct calculation gives
\begin{equation*}
\rho D_t e_F + s \rho D_t \theta - S_\Gamma (v , \sigma ) : D_\Gamma (v) = - \tilde{e}_D \text{ on }\Gamma_T.
\end{equation*}
Therefore we have
\begin{equation}\label{Eq510}
\begin{cases}
\rho D_t h = {\rm{div}}_\Gamma \{ e'_3 (|{\rm{grad}}_\Gamma \theta |^2) {\rm{grad}}_\Gamma \theta \} + \tilde{e}_D + D_t \sigma,\\
\theta \rho D_t s = {\rm{div}}_\Gamma \{ e'_3 ( |{\rm{grad}}_\Gamma \theta |^2) {\rm{grad}}_\Gamma \theta \} + \tilde{e}_D,\\
\rho D_t e_F + s \rho D_t \theta - S_\Gamma (v , \sigma ) : D_\Gamma (v) = - \tilde{e}_D.
\end{cases}
\end{equation}
We easily check that \eqref{Eq510} satisfies \eqref{eq13} and \eqref{eq14}.

Next, we study the conservation and energy laws of system \eqref{eq11}. We assume that for $0 < t <T$
\begin{equation}\label{Eq511}
\frac{\partial \theta }{\partial \nu} \bigg|_{\partial \Gamma (t)} = 0,{ \ }\frac{\partial C }{\partial \nu} \bigg|_{\partial \Gamma (t)} = 0,{ \ }S_\Gamma (v , \sigma ) \nu \bigg|_{\partial \Gamma (t)} = 0,
\end{equation}
where ${\partial f}/{\partial \nu} = (\nu \cdot \nabla_\Gamma ) f$. Since
\begin{align*}
\frac{d}{d t} \int_{\Gamma (t)} \rho v { \ }d \mathcal{H}^2_x & = \int_{\Gamma (t)} \rho D_t v { \ }d \mathcal{H}_x^2,\\
\frac{d}{d t} \int_{\Gamma (t)} E_S { \ }d \mathcal{H}^2_x & = \int_{\Gamma (t)} \{ \rho D_t v \cdot v + \rho D_t e  \} { \ }d \mathcal{H}_x^2,\\
\frac{d}{d t} \int_{\Gamma (t)} C { \ }d \mathcal{H}^2_x & = \int_{\Gamma (t)} \{ D_t C + ({\rm{div}}_\Gamma v ) C  \}{ \ }d \mathcal{H}_x^2
\end{align*}
from the surface transport theorem and \eqref{eq51}, we use system \eqref{eq11}, \eqref{eq29}, and \eqref{Eq511} to see that \eqref{eq19}, \eqref{Eq111}, and \eqref{Eq112} hold for $t_1 < t_2$.

Finally, we investigate the law of conservation of angular momentum. Since $D_t x = 2 v$ and $v \times v = 0$, we use the surface transport theorem and \eqref{eq11} to find that for $t_1 < t_2$
\begin{multline*}
\int_{\Gamma (t_2)} x \times ( \rho v ) { \ }d \mathcal{H}^2_x = \int_{\Gamma (t_1)} x \times (\rho v) { \ }d \mathcal{H}^2_x + \int_{t_1}^{t_2} \int_{\Gamma (\tau)} x \times ( \rho F ) { \ }d \mathcal{H}^2_x d \tau\\
 + \int_{t_1}^{t_2} \int_{\Gamma (\tau)} x \times \{ {\rm{div}}_\Gamma S_\Gamma (v , \sigma ) \}  { \ }d \mathcal{H}^2_x d \tau .
\end{multline*}
If we prove that
\begin{equation}\label{Eq512}
\int_{\Gamma ( t )} x \times \{ {\rm{div}}_\Gamma S_\Gamma (v , \sigma ) \}  { \ }d \mathcal{H}^2_x = { }^t ( 0, 0 , 0),
\end{equation}
then we have the law of conservation of angular momentum \eqref{Eq110}. Now we show \eqref{Eq512}. Set $M = S_\Gamma (v , \sigma )$, that is, $[M]_{ij} = [S_\Gamma (v, \sigma )]_{i j}$. A direct calculation gives\footnotesize
\begin{multline*}
x \times {\rm{div}}_\Gamma M = \\
\begin{pmatrix}
x_2 ( \partial_1^\Gamma [M]_{31} + \partial_2^\Gamma [M]_{32} + \partial_3^\Gamma [M]_{33} ) - x_3 ( \partial_1^\Gamma [M]_{21} + \partial_2^\Gamma [M]_{22} + \partial_3^\Gamma [M]_{23} ) \\
x_3 ( \partial_1^\Gamma [M]_{11} + \partial_2^\Gamma [M]_{12} + \partial_3^\Gamma [M]_{13} ) - x_1 ( \partial_1^\Gamma [M]_{31} + \partial_2^\Gamma [M]_{32} + \partial_3^\Gamma [M]_{33} ) \\
x_1 ( \partial_1^\Gamma [M]_{21} + \partial_2^\Gamma [M]_{22} + \partial_3^\Gamma [M]_{23} ) - x_2 ( \partial_1^\Gamma [M]_{11} + \partial_2^\Gamma [M]_{12} + \partial_3^\Gamma [M]_{13} ) 
\end{pmatrix}.
\end{multline*}\normalsize
We shall prove that for each $i,j =1,2,3,$\footnotesize
\begin{equation*}
\int_{\Gamma (t)} \{ x_i ( \partial_1^\Gamma [M]_{j 1} + \partial_2^\Gamma [M]_{j 2} + \partial_3^\Gamma [M]_{j 3} ) - x_j ( \partial_1^\Gamma [M]_{i 1} + \partial_2^\Gamma [M]_{i 2} + \partial_3^\Gamma [M]_{i 3} ) \} { \ }d \mathcal{H}^2_x = 0.
\end{equation*}\normalsize
Using the integration by parts \eqref{eq29} and the facts that $S_\Gamma (v, \sigma) n = { }^t (0,0,0)$ on $\Gamma (t)$ and $S_\Gamma (v,\sigma ) \nu = { }^t (0,0,0)$ on $\partial \Gamma (t)$, we see that for each $i,j =1,2,3,$
\begin{equation*}
\int_{\Gamma (t)} \{ x_i ( \partial_1^\Gamma [M]_{j 1} + \partial_2^\Gamma [M]_{j 2} + \partial_3^\Gamma [M]_{j 3} ) { \ }d \mathcal{H}^2_x = - \int_{\Gamma (t)} [M]_{j i} { \ }d \mathcal{H}^2_x.
\end{equation*}
Since $[M]_{j i } = [M]_{i j}$, we check that\footnotesize
\begin{multline*}
\int_{\Gamma (t)} \{ x_i ( \partial_1^\Gamma [M]_{j 1} + \partial_2^\Gamma [M]_{j 2} + \partial_3^\Gamma [M]_{j 3} ) 
- x_j ( \partial_1^\Gamma [M]_{i 1} + \partial_2^\Gamma [M]_{i 2} + \partial_3^\Gamma [M]_{i 3} ) \} { \ }d \mathcal{H}^2_x\\
= - \int_{\Gamma (t)} [M]_{j i} { \ }d \mathcal{H}^2_x + \int_{\Gamma (t)} [M]_{i j} { \ }d \mathcal{H}^2_x = 0.
\end{multline*}\normalsize
This implies \eqref{Eq512}. Therefore Theorem \ref{Thm210} is proved. From \eqref{eq53} and \eqref{eq55}, we see the assertion $(\mathrm{i})$ of Theorem \ref{Thm211}.

\subsection{Tangential Compressible Fluid System}\label{subsec53}

Let us derive the generalized tangential compressible fluid system \eqref{Eq113} on the evolving surface $\Gamma (t)$ with a boundary. We assume that $v \cdot n = 0$ on $\Gamma (t)$. We set the generalized energy densities for compressible fluid on the surface as in Assumption \ref{ass12}, and the action integral $Act$, the energy $E_D$ dissipation due to the viscosities, the work $E_W$ done by the pressure and exterior force, the energy $E_{TD}$ dissipation due to thermal diffusion, and the energy $E_{GD}$ due to general diffusion as in Section \ref{sect2}. Based on Proposition \ref{prop28}, we have \eqref{eq51}.

We first derive the momentum equation of our tangential compressible fluid system. From Theorems \ref{thm27} and \ref{thm29}, we have the following forces:
\begin{align*}
\frac{\delta Act}{\delta \widetilde{x}} \bigg|_{z \cdot n =0} & = P_\Gamma \rho D_t v,\\
\frac{\delta E_{D + W}}{\delta v} \bigg|_{\varphi \cdot n =0 } & = P_\Gamma {\rm{div}}_\Gamma S_\Gamma (v , \sigma ) + P_\Gamma \rho F .
\end{align*}
Here $E_{D + W} = E_D + E_W$. We assume the following energetic variational principle:
\begin{equation*}
\frac{\delta Act}{\delta \widetilde{x}} \bigg|_{z \cdot n =0} = \frac{\delta E_{D + W}}{\delta v}\bigg|_{\varphi \cdot n =0 }
\end{equation*}
to have
\begin{equation}\label{Eq513}
P_\Gamma \rho D_t v = P_\Gamma {\rm{div}}_\Gamma S_\Gamma ( v , \sigma ) + P_\Gamma \rho F.
\end{equation}
Note that we may use $A_B$ instead of $Act$.

Secondly, we consider the internal energy. By an argument similar to that in subsection \ref{subsec51}, we obtain \eqref{eq53}. Here $\tilde{e}_D$ and $B$ are defined by \eqref{eq54} and \eqref{eq55}. Notice that $P_\Gamma v = v$ and $v \cdot v = P_\Gamma v \cdot v$. Now we assume that $B \equiv 0$ and $F \equiv 0$. Then we have an energy equality:
\begin{equation*}
\int_{\Gamma (t_2)} \frac{1}{2} \rho | v |^2 { \ }d \mathcal{H}^2_x + \int_{t_1}^{t_2} \int_{\Gamma ( \tau )} \{ \tilde{e}_D - ( {\rm{div}}_\Gamma v ) \sigma \} { \ }d \mathcal{H}^2_x d \tau = \int_{\Gamma (t_1)} \frac{1}{2} \rho | v |^2 { \ }d \mathcal{H}^2_x. 
\end{equation*}
This implies that $\tilde{e}_D - ( {\rm{div}}_\Gamma v ) \sigma$ is the density for the dissipation energy and the work done by pressure of the our fluid system. From Theorem \ref{thm27} we have \eqref{eq56} and \eqref{eq57}. Applying the first law of thermodynamics, we have \eqref{eq58}.

Finally, we derive the generalized diffusion system. We assume that the change of rate of the concentration $C$ equals to the force derived from a variation of the energy dissipation due to general diffusion. Then we have \eqref{eq59}.

Combining \eqref{eq51}, \eqref{Eq513}, \eqref{eq58}, and \eqref{eq59}, we therefore have the tangential compressible fluid system \eqref{Eq113}. Notice that $D_t f = D_t^\Gamma f $ since $v \cdot n =0$. From \eqref{eq53} and \eqref{eq55} the assertion $(\mathrm{ii})$ of Theorem \ref{Thm211} is proved.

\subsection{Barotropic Compressible Fluid Systems}\label{subsec54}

Based on Proposition \ref{prop28}, we obtain \eqref{eq51}. From the assertion $(\mathrm{ii})$ of Theorem \ref{thm29} we obtain systems \eqref{Eq116} and \eqref{Eq117}. The proof of the assertions $(\mathrm{iii})$ and $(\mathrm{iv})$ in Theorem \ref{Thm211} are left to the reader.

\end{document}